\documentclass[letterpaper]{article} 
\usepackage{aaai25}  
\usepackage{times}  
\usepackage{helvet}  
\usepackage{courier}  
\usepackage[hyphens]{url}  
\usepackage{graphicx} 
\usepackage{natbib}  
\usepackage{caption} 
\usepackage{amsmath,amssymb,amsthm}
\usepackage[ruled,vlined,linesnumbered]{algorithm2e}
\usepackage{booktabs}

\newtheorem{proposition}{Proposition} 
\newtheorem{theorem}{Theorem}
\newtheorem{lemma}{Lemma}
\newtheorem{corollary}{Corollary}
\newtheorem{definition}{Definition}
\newtheorem{remark}{Remark}

\def\Reals{\mathbb{R}}
\def\NonNegativeReals{\Reals_+}
\def\PositiveReals{\Reals_{++}}
\def\States{\mathbb{X}}
\def\Controls{\mathbb{U}}
\def\Universum{\Omega}
\def\StatesTarget{\States_{\star}}
\def\StatesUnsafe{\States_{\oslash}}
\def\StatesInitial{\States_{0}}
\def\EventRA{E_{\text{RA}}}
\def\EventS{E_{\text{S}}}
\def\batch{\mathbb{B}}

\def\Algebra{\mathcal{F}}
\def\NaturalAlgebra{\mathcal{N}}
\def\StatesTargetFamily{\mathcal{X}_{\star}}
\def\StatesUnsafeFamily{\mathcal{X}_{\oslash}}
\def\StatesInitialFamily{\mathcal{X}_{0}}
\def\Cells{\mathcal{C}}
\def\UnverifiedCells{\mathcal{U}}
\newcommand{\BorelAlgebra}[1]{\mathcal{B}(#1)}
\newcommand{\Polish}[1]{\bigl(#1, \BorelAlgebra{#1}\bigr)}

\def\NaturalFiltration{\mathfrak{N}}
\def\filtration{\mathfrak{F}}

\def\drift{f}
\def\diffusion{g}
\def\policy{\pi}
\def\certificate{V}
\def\generatorBoundRA{\zeta}
\def\generatorBoundS{\xi}

\def\WienerProcess{W}
\def\stateProcess{X}

\newcommand{\trajectoryProcess}[4]{\stateProcess^{#2,#3}_{#1,#4}}

\newcommand{\operatorStyle}[1]{\mathsf{#1}}
\def\dd{\mathrm{d}} 
\def\generator{\operatorStyle{G}}

\newcommand{\Hessian}[1]{\operatorStyle{H}_{#1}}
\newcommand{\gradient}[1]{\nabla_{\!#1}}
\DeclareMathOperator{\expect}{\mathbb{E}}
\DeclareMathOperator{\probability}{\mathbb{P}}
\DeclareMathOperator{\interior}{\operatorStyle{int}}
\DeclareMathOperator{\diag}{\operatorStyle{diag}}
\newcommand{\indicator}[1]{\operatorStyle{I}_{#1}}

\def\theStoppingTime{\psi}

\def\probRA{\varepsilon}
\def\probS{\delta}
\def\upperRAbound{\beta_{\text{RA}}}
\def\lowerRAbound{\alpha_{\text{RA}}}
\def\upperSbound{\beta_{\text{S}}}
\def\lowerSbound{\alpha_{\text{S}}}

\def\estProbRA{\hat\probRA}
\def\estProbS{\hat\probS}
\def\estUpperRAbound{\hat\beta_{\text{RA}}}
\def\estLowerRAbound{\hat\alpha_{\text{RA}}}
\def\estUpperSbound{\hat\beta_{\text{S}}}
\def\estLowerSbound{\hat\alpha_{\text{S}}}

\def\placeholder{{\,\cdot\,}}
\newcommand{\define}[2]{#1 \triangleq #2}

\def\pndAngle{\theta}
\def\pndVelocity{\varphi}
\def\pndGravity{g}
\def\pndMass{m}
\def\pndFriction{b}
\def\pndLength{L}
\def\pndMaxTorque{M}

\title{Neural Continuous-Time Supermartingale Certificates}
\author{
    Grigory Neustroev\textsuperscript{\rm 1},
    Mirco Giacobbe\textsuperscript{\rm 2},
    Anna Lukina\textsuperscript{\rm 1}
}
\affiliations {
    \textsuperscript{\rm 1}Delft University of Technology, the~Netherlands\\
    \textsuperscript{\rm 2}University of Birmingham, UK\\
    g.neustroev@tudelft.nl, m.giacobbe@bham.ac.uk, a.lukina@tudelft.nl
}

\begin{document}

\maketitle

\begin{abstract}
We introduce for the first time a neural-certificate framework for continuous-time stochastic dynamical systems.
Autonomous learning systems in the physical world demand continuous-time reasoning, yet existing learnable certificates for probabilistic verification assume discretization of the time continuum.
Inspired by the success of training neural Lyapunov certificates for deterministic continuous-time systems and neural supermartingale certificates for stochastic discrete-time systems, we propose a framework that bridges the gap between continuous-time and probabilistic neural certification for dynamical systems under complex requirements.
Our method  combines machine learning and symbolic reasoning to produce formally certified bounds on the probabilities that a nonlinear system satisfies specifications of reachability, avoidance, and persistence.
We present both the theoretical justification and the algorithmic implementation of our framework and showcase its efficacy on popular benchmarks.
\end{abstract}

\begin{links}
    \link{Code}{https://doi.org/10.5281/zenodo.14537003}
\end{links}

\section{Introduction}

Ensuring safety is paramount to the alignment and governance of autonomous learning systems.
Providing verifiable guarantees of compliance with safety requirements is essential for building trust between a system, its users, and the regulators.
This is especially critical when the system is deployed in the physical world, where it can cause harm.
A central challenge in designing autonomous learning systems is developing control policies that ensure the desired spatial and temporal behavior.
These objectives include reaching a specific target region (reachability), consistently avoiding unsafe regions (avoidance), reaching and remaining within a safe region (persistence), and combinations thereof.

Autonomous systems often face multiple probabilistic uncertainties, leading to noisy model parameters or stochastic system evolution.
Safe control design for the most general models---nonlinear continuous-time stochastic dynamical systems---provides \emph{formal} guarantees for systems operating in uncertain physical environments.
Powerful techniques for achieving this are based on Lyapunov functions, which serve as a potential-energy characterization for differential equations and, under certain conditions, become formal proof certificates for desired behavioral properties.
Initially introduced for the asymptotic stability analysis of deterministic systems with respect to an equilibrium \cite{kalman1960control}, Lyapunov functions have been generalized to stochastic dynamical systems \cite{blumenthal1968markov} and extended into formal certificates for reachability, avoidance (known as barrier certificates), and persistence properties \cite{prajna2004stochastic}.
 
\begin{figure}[t]
    \centering
    \includegraphics[width=\columnwidth]{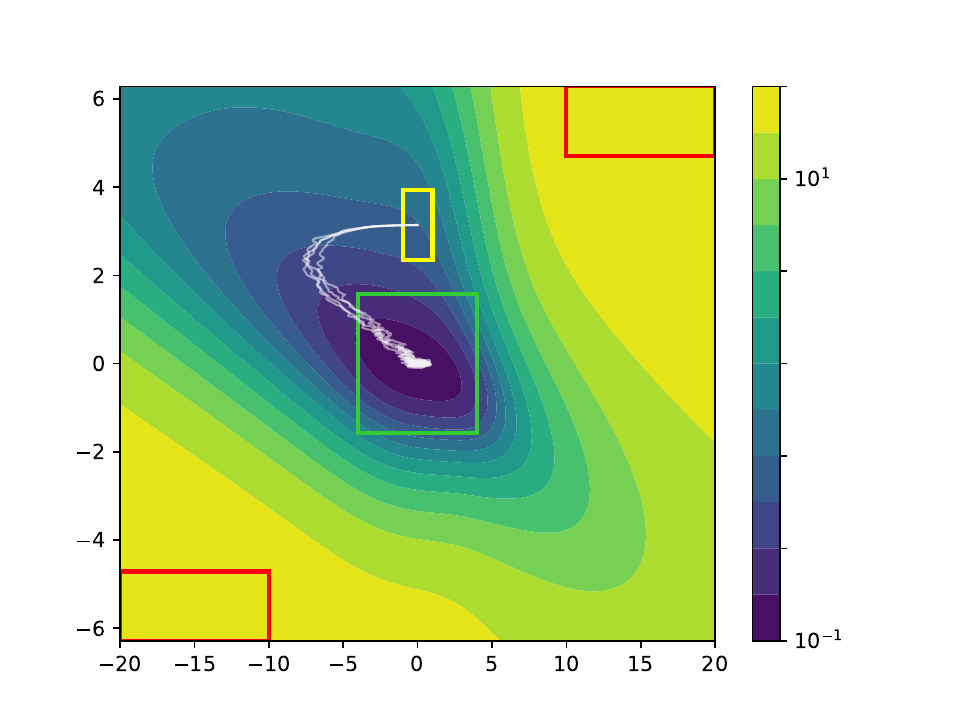}
    \caption{%
      A neural supermartingale certificate for the continuous-time stochastic inverted pendulum.
      Darker colors indicate higher probability for the system trajectories (sampled in white) to reach and remain in the green rectangle, while avoiding the red rectangles.%
    }
    \label{fig:certificate-example}
\end{figure}

Computing proof certificates analytically is a hard problem for nonlinear systems, which has traditionally required manual effort and considerable expertise.
This approach is feasible when a system  design is familiar to an engineer.
However, for autonomous learning systems with embedded neural networks, whose internals are often not human-comprehensible, the construction of proof certificates must be automated.
The best-known algorithms for their construction are based on sum-of-square programming \cite{papachristodoulou2002construction,topcu2008local},
and are limited to polynomial systems.
Systems with dynamics involving trigonometric, exponential, or piecewise functions, which include neural building blocks, are beyond their reach.
Therefore, systems that embed neural networks necessarily require more flexible techniques.

Neural networks as control policies for nonlinear dynamical systems have long proven extreme efficacy when in the presence of complex tasks.
This has inspired the introduction of Lyapunov certification into learning \cite{richards2018lyapunov,dawson2023safe}, and the growing demand for formally verified AI has motivated their coupling with symbolic reasoning to guarantee soundness of the certificate \cite{chang2019neural,abate2020formal}.
These methods exploit the expressive capabilities of neural networks to represent proof certificates---referred to as \emph{neural certificates}---and build upon the fact that algorithmically verifying neural networks is increasingly  feasible~\cite{DBLP:conf/cav/KatzBDJK17,DBLP:conf/sp/GehrMDTCV18,DBLP:conf/cav/WuIZTDKRAJBHLWZKKB24}.
This is particularly true in control applications, where the neural networks involved are typically much smaller compared to those used in computer vision or natural language processing~\cite{DBLP:conf/cav/TranYLMNXBJ20,DBLP:conf/ijcai/BacciG021,DBLP:journals/aamas/AkintundeBKL22}.

Neural certificates for deterministic continuous-time dynamical systems implement Lyapunov-like conditions that are both amenable to gradient descent for learning and compatible with symbolic reasoning for formal verification.
While Lyapunov conditions are well-defined for deterministic systems and worst-case reasoning, they fall short in quantifying the probability of compliance with a requirement for stochastic systems.
Moreover, worst-case reasoning may consider measure-zero trajectories as plausible, imposing sure satisfaction of a requirement, which is unreasonably conservative in many practical scenarios.
Due to the fundamental differences between (deterministic) Newton--Leibnitz calculus and (stochastic) It\^o calculus, the safety assurance of stochastic dynamical systems demands quantitative probabilistic reasoning techniques.

Probabilistic reasoning for the safety assurance of stochastic systems is grounded in martingale theory, which establishes conditions for their compliance with behavioral requirements \cite{prajna2007framework}.
More specifically, proof rules for probabilistic verification rely on the construction of supermartingale processes.
Neural certificates based on supermartingale-like conditions have been successful in proving termination for probabilistic programs~\cite{abate2021learning,DBLP:conf/concur/AbateEGPR23} and, analogously, probabilistic reachability, avoidance, and reach-avoidance for discrete-time stochastic dynamical systems \cite{ansaripour2023learning, DBLP:journals/csysl/MathiesenCL23,vzikelic2023learning,badings2024learning}.
However, all existing results on \emph{neural supermartingale certificates} assume discretization of the time continuum, which is suitable for computer simulations of dynamical systems and experiments \textit{in silico}, but provide no formal guarantees for systems in continuous time.

We present for the first time a framework for the construction of neural supermartingale certificates in continuous time, with formal guarantees for reachability, avoidance and persistence requirements.
We provide 
(1) a quantitative proof rule for continuous-time reach-avoid-stay properties, 
(2) an algorithm for training neural supermartingale certificates from samples of drift and diffusion of the stochastic dynamics, and
(3) a symbolic technique for the formal verification of our neural certificates based on interval-bound propagation.
Our method produces a quantitative neural certificate that yields a sound lower bound on the probability of compliance of the system with the given requirement.

Figure~\ref{fig:certificate-example} illustrates a valid quantitative neural certificate for a stochastic inverted pendulum subject to a neural controller, depicting its likelihood of avoiding the red rectangles while reaching and remaining within the green rectangle.
It can be observed that any trajectory initialized in a given state satisfies the property with at least the probability indicated for that state.
Specifically, our certificate shows that every trajectory initialized within the yellow rectangle satisfies the reach-stay-avoidance property with high probability.

We demonstrate the effectiveness of our method on well-known continuous-time stochastic system models from the literature.
Our examples include systems with dynamics that incorporate trigonometric functions and neural controllers, which are beyond the capabilities of traditional algorithms for provably safe control of continuous-time stochastic dynamical systems.
Our results establish a foundation for unsupervised and practically effective safeguarded stochastic control based on neural certificates, offering quantitative guarantees for continuous-time dynamical models.

In summary, our contribution is twofold.
First, we introduce a comprehensive methodology for formally guaranteed neural reach-avoid-stay certification of continuous-time stochastic dynamical systems controlled by a given policy.
Second, we developed a prototype based on our technique and demonstrated its practical effectiveness on examples that are beyond the capabilities of traditional methods for automated provably safe control.

\section{Problem Statement}

Consider a probability space \((\Universum, \Algebra,  \probability)\) with filtration \(\filtration = (\Algebra_t)_{t \geq 0}\).
Let \(\Polish{\Reals^l}\) and \(\Polish{\Controls}\) be Polish spaces of states and actions (here \(\BorelAlgebra{\placeholder}\) denotes the Borel \(\sigma\)-algebra, i.e., the smallest \(\sigma\)-algebra that contains the open intervals).
In other words, we assume that the states are real-valued vectors of length \(l < \infty\), but controls can be anything Borel-measurable.

Consider a continuous-time dynamical system described by the following stochastic differential equation (SDE):
\begin{equation}
  \dd \stateProcess_t
  = \drift\bigl(t, \stateProcess_t, u\bigr)\,\dd t
  + \diffusion\bigl(t, \stateProcess_t, u\bigr)\,\dd \WienerProcess_t
  \label{eq:general_sde}
\end{equation}
where \(\drift: \NonNegativeReals \times \Reals^l \times \Controls \to \Reals^{l}\) and \(\diffusion: \NonNegativeReals \times \Reals^l \times \Controls \to \Reals^{l} \times \Reals^{k}\) are vector- and matrix-valued \emph{drift} and \emph{diffusion} functions, \(\define{\NonNegativeReals}{[0, \infty)}\);
they represent the deterministic and stochastic parts of the dynamics.
\((\WienerProcess_t)_{t\geq 0}\) is a \(k\)-dimensional Wiener process with independent coordinates.

We assume that the state space \(\Reals^l\) contains a subset \(\States\) of interest that we restrict our attention to.
For example, the system can be known to remain in \(\States\) \emph{almost surely (a.s.)}, that is, with probability one.

The control signal \(u = \policy\bigl(t, \stateProcess_t\bigr)\) is given by a function of time and state known as a \emph{policy} \(\policy : \NonNegativeReals \times \Reals^l \to \Controls\).
When a policy is fixed, we denote
\begin{align*}
  \define{\drift_\policy(t, x)&}{\drift\bigl(t, x, \policy(t, x)\bigr),}
  \define{&\diffusion_\policy(t, x)&}{\diffusion\bigl(t, x, \policy(t, x)\bigr).}
\end{align*}

We use \((\trajectoryProcess{\policy}{s}{x}{t})_{t\geq s}\) for a stochastic state process issuing from a state \(x \in \Reals^l\) at time \(s \in \NonNegativeReals\) when the system is controlled by a policy \(\policy\).
By its definition,
\begin{equation}\label{eq:trajectory-origin}
  \trajectoryProcess{\policy}{t}{x}{t} \equiv x.
\end{equation}

For a given policy \(\policy\), our goal is to design an algorithm for verifying that the system's stochastic dynamics satisfy the conditions of the following definition.

\begin{definition}\label{df:reach-avoid-specification}
A \emph{probabilistic reach-stay-avoid (RAS) specification} is a tuple \((\StatesTargetFamily, \StatesUnsafeFamily, \StatesInitial, \probRA, \probS)\), where 
\begin{itemize}
  \item \(\define{\StatesTargetFamily}{\bigl(\StatesTarget(t)\bigr)_{t\geq 0}}\) is a family of target sets \(\StatesTarget(t) \subset \States\),
  \item \(\define{\StatesUnsafeFamily}{\bigl(\StatesUnsafe(t)\bigr)_{t\geq 0}}\) is a family of unsafe sets \(\StatesUnsafe(t) \subset \States\),
  \item \(\StatesInitial\) is a set of initial states,
  \item \(\probRA \in [0,1)\) is an avoidance probability, and
  \item \(\probS  \in [0,1)\) is a stay probability.
\end{itemize}
A policy \(\policy\) \emph{satisfies} this specification if for every initial state \(x_0 \in \StatesInitial\) there exists a.s. a finite time \(\tau < \infty\) such that the probability to reach the target set without ever entering the unsafe set by the time \(\tau\) is at least \(\probRA\), and the probability to remain in the target set afterwards is at least \(\probS\);
in other words, if \(\probability[\EventRA] \geq \probRA\) and \(\probability[\EventS] \geq \probS\) for the events
\begin{align}
  \define{\EventRA &}{\bigl\{\trajectoryProcess{\policy}{0}{x_0}{\tau} \in \StatesTarget(\tau) \wedge \forall t < \tau: \trajectoryProcess{\policy}{0}{x_0}{t} \notin \StatesUnsafe(t)\bigr\},}\label{eq:reach-avoid}\\
  \define{\EventS &}{\bigl\{\forall t \geq \tau: \trajectoryProcess{\policy}{0}{x_0}{t} \in \StatesTarget(t)\bigr\}.}\label{eq:stay}
  \end{align}
\end{definition}

The dependence of the initial, target, and unsafe sets on time allows us to tackle problems that are not time-homogenous.
For example, the unsafe set \(\StatesUnsafe(t)\) can move or grow over time or the policy can evolve during training.

\subsection{Model Assumptions \& Further Notation}

A policy is assumed to be \emph{admissible}, that is, the coefficients \(\drift_\policy(t, x)\) and \(\diffusion_{\policy, r}(t, x)\) are Lipschitz-continuous, have continuous derivatives with respect to \(x\) which are bounded uniformly in \(t > 0\).
We use \(\diffusion_{\policy, r}, r = 1, \dots, k\) to denote the column vectors of \(\diffusion_\policy\).

The set \(\States\) is compact.
For \(\States_\circledast \in \{\StatesTarget, \StatesUnsafe\}\) the sets
\begin{equation}\label{eq:time-extended-sets}
  \define{\mathcal{X}_\circledast}{\bigl\{(t, x) \bigm| t \geq 0 \wedge x \in \States_\circledast(t) \bigr\}}
\end{equation}
are Borel.
With a slight abuse of notation, we denote these sets \(\StatesTargetFamily\) and \(\StatesUnsafeFamily\).
The initial set \(\StatesInitial\) is compact and belongs to the interior of safe states at time \(0\), \(\StatesInitial \subset \interior \bigl(\States \setminus \StatesUnsafe(0)\bigr)\).
The sets \(\StatesTarget(t)\) belong to the interiors of safe sets, \(\StatesTarget(t) \subset \interior \bigl(\States \setminus \StatesUnsafe(t)\bigr)\).

Without loss of generality, we assume that the equilibrium of the system is at zero, \(x_\star \equiv 0\) and always belongs to the set \(\StatesTarget(t)\);
similarly, \(u_\star \equiv 0\).
Moreover, \(\policy(t,0) \equiv 0\), \(\drift_\policy(t,0) \equiv 0\), and \(\diffusion_{\policy,r}(t,0) \equiv 0\).

The design of our certificate defines its properties within different sublevel sets.
A \emph{sublevel set} \(L^-_\gamma(\certificate)\) of level \(\gamma\) (a \emph{sub-\(\gamma\) set}) of a function \(\certificate(t, x)\) is defined as
  \[
    \define%
    {L^-_\gamma(\certificate)}%
    {\bigl\{(t, x) \in \NonNegativeReals \times \States \bigm| \certificate(t, x) \leq \gamma \bigr\}.}
  \]

\begin{remark}
  Note that we restrict sublevel sets to subsets of a bounded domain \(\States\) ensuring that they are bounded as well.
\end{remark}

\subsection{A Primer on Stochastic Calculus}

Under our assumptions, the problem \eqref{eq:general_sde} has a solution that is a so-called Feller--Dynkin process\footnote{This roughly means that it is right-continuous, exhibits the strong Markov property, and has time-independent transitions.}.
The dynamics of a Feller--Dynkin process can be described by a linear operator \(\generator_\policy\) called its \emph{infinitesimal generator}, which in our case admits the following form \cite{khasminskii2011stability}:
\begin{equation}\label{eq:generator}
  \define{\generator_\policy}{ \tfrac{\partial}{\partial t} + \sum_{i=1}^l \drift_\policy(t, x_i)\tfrac{\partial}{\partial x_i} + \tfrac{1}{2} \sum_{i=1}^l\sum_{r=1}^k \diffusion_{\policy, r}^2(t, x_i) \tfrac{\partial^2}{\partial x_i^2}.}
\end{equation}
Infinitesimal generator is the stochastic counterpart of gradient, describing the direction of fastest expected increase.

Next, let us recall the definition of a natural filtration and stopping time.
The \emph{natural filtration} of \(\Algebra\) with respect to a stochastic process \((\eta_t)_{t\geq 0}\) is the smallest \(\sigma\)-algebra \((\NaturalAlgebra_t)_{t\geq 0}\) on \(\Universum\) that contains all pre-images of \(\BorelAlgebra{\States}\)-measurable subsets of \(\States\) for times \(s\) up to \(t\),
\[\define{\NaturalAlgebra_t}{\sigma \bigl\{\eta_s^{-1}(A) \bigm| s \in \NonNegativeReals \wedge s \leq t \wedge A \in \BorelAlgebra{\States}\bigr\}.}\]
The natural filtration represents the information flow generated by a stochastic process:
at any moment of time \(t\), the \(\sigma\)-algebra \(\NaturalAlgebra_t\) includes all of the information generated by the process up to that point of time.

A random variable \(\tau: \Universum \to \NonNegativeReals\) is called an \emph{\(\filtration\)-stopping time} if
  \(\{\tau \leq t\} \in \Algebra_t \quad \text{for all }t\in\NonNegativeReals.\)
Intuitively, \(\tau\) is a stopping time if it is possible to determine whether it happened without looking into the future.
We can ``freeze'' a stochastic process at a stopping time (hence the name).
Formally, this results in a new version of a stochastic process called a stopped process that is defined for all \(\omega\in\Universum\) as
  \[\define{\eta_{t \wedge \tau}(\omega)}{\eta_{t}(\omega) \indicator{\{t \leq \tau(\omega)\}} + \eta_{\tau(\omega)}(\omega) \indicator{\{t > \tau(\omega)\}}}.\]
The stopped process follows the original process up to a random time \(\tau\), and after that point it stays the same.
Stopped processes are of special interest to us because they are used to construct supermartingales.

An \(\filtration\)-adapted (i.e., such that \(\eta_t\) is always \(\Algebra_t\)-measurable) stochastic process \((\eta_t)_{t\geq 0}\) is called an \emph{\(\filtration\)-supermartingale} if \(\expect[\eta_t\,|\,\Algebra_s] \leq \eta_s\) for every pair of times \(t > s \geq 0\).

\subsection{Comparison to Alternative Models}

Continuous-time deterministic certificates do not consider the quadratic term in \eqref{eq:generator} and therefore fail to fully capture the stochastics of a system defined by \eqref{eq:general_sde}.
For example, the solutions to \(\dot{x} = x\) and \(\dd \stateProcess_t = \stateProcess_t \dd t + \stateProcess_t\,\dd \WienerProcess_t\) are \(x_t = x_0 e^t\) and \(X_t = x_0 e^{t/2 + \WienerProcess_t}\); the growth rates (\(1\) and \(1/2\)) are different.
Similarly, discrete-time stochastic systems represent the dynamics of a continuous-time system only approximately. The solutions to \(x_{t+1} - x_{t} = -\tfrac{1}{2}x_{t} + \eta_t\) where \(\eta_t\sim \mathcal{N}(0,1)\) (i.i.d.) and \(\dd \stateProcess_t = -\tfrac{1}{2}\stateProcess_t \dd t + \dd \WienerProcess_t\) are \(x_{t} = 2^{-t} x_0 + \sum_{i=1}^{t} 2^{i-t} \eta_{i-1} = e^{-t \ln 2} x_0 + \tfrac{2\sqrt{3}}{3}\sqrt{1-4^{-t}} \tilde{\eta}_t\), \(\tilde\eta_t\sim \mathcal{N}(0,1)\) and \(X_t = e^{-t/2}x_0 + \WienerProcess_{1-e^{-2t}}\); again, the decay rates (\(\ln 2\) and \(0.5\)) differ, leading to different dynamics.

These examples show that existing certificates cannot be directly applied to the continuous-time stochastic setting.
To satisfy a RAS specification, we need to find a new type of certificate and prove its validity for our problem.

\section{Reach-Avoid-Stay Certificates}

\begin{definition}\label{df:ras-C}
  A function \(\certificate(t, x) : \NonNegativeReals \times \Reals^l \to \Reals\) is called a \emph{reach-avoid-stay certificate (RAS-C)} if in the domain \(\PositiveReals \times \States\) it is bounded and twice continuously differentiable with respect to \(x\) and continuously differentiable with respect to \(t\), and, moreover, the following conditions hold, given the sets \(\States\) and \(\StatesInitial\), the families of sets \(\StatesUnsafeFamily\), and \(\StatesTargetFamily\), as well as some positive constants \(\lowerSbound < \upperSbound < \lowerRAbound < \upperRAbound\):
  \begin{enumerate}
    \item\label{cnd:nonneg} \emph{nonnegativity:} \(\certificate(t, x) \geq 0\) for all \(t \geq 0\) and \(x \in \States\);
    \item\label{cnd:init} \emph{initial condition:} \(\certificate(0, x_0) \leq \lowerRAbound\) for all \(x_0 \in \StatesInitial\);
    \item\label{cnd:safe} \emph{safety condition:} \(\certificate(t, x) \geq \upperRAbound\) for all \((t, x) \in \StatesUnsafeFamily\);
    \item\label{cnd:decr} \emph{decrease condition:} there exists \(\generatorBoundRA(t)\) satisfying
        \begin{equation}\label{eq:zeta-integral}
          \lim_{t \to \infty}\int_0^t \generatorBoundRA(s)\,\dd s = \infty\quad\text{and}\quad\generatorBoundRA(t) > 0,
        \end{equation}
        and
        \(
          \generator_\policy \certificate(t, x) \leq -\generatorBoundRA(t) 
        \)
        for all \((t, x) \in L^-_{\upperRAbound}(\certificate)\setminus \interior\StatesTargetFamily\).
    \item\label{cnd:goal} \emph{goal condition:} in the target set \(\StatesTargetFamily\) there exists a sub-\(\upperSbound\) set, \(L^-_{\upperSbound}(\certificate)\);
    \item\label{cnd:stay} \emph{stay condition:} there exists \(\generatorBoundS(t)\) satisfying
        \begin{equation}\label{eq:gamma-integral}
          \lim_{t \to \infty}\int_0^t \generatorBoundS(s)\,\dd s = \infty\quad\text{and}\quad\generatorBoundRA(t) > 0,
        \end{equation}
        and
        \(
          \generator_\policy \certificate(t, x) \leq -\generatorBoundS(t)
        \)
        for all \((t, x) \in \bigl(\StatesTargetFamily \cap L^-_{\upperRAbound}(\certificate)\bigr) \setminus \interior L^-_{\lowerSbound}(\certificate)\).
  \end{enumerate}
\end{definition}

\begin{remark}
  Conditions \eqref{eq:zeta-integral} and \eqref{eq:gamma-integral} are satisfied by any positive constants \(\generatorBoundRA\) and \(\generatorBoundS\).
  We use a more general definition, because sometimes it is possible to derive better bounds based on the domain knowledge of the problem at hand.
\end{remark}

At the core of the definition is the following idea: every process issuing within the sub-\(\lowerRAbound\) set \(L^-_{\lowerRAbound}(\certificate)\) is unlikely to leave the sub-\(\upperRAbound\) set \(L^-_{\upperRAbound}(\certificate)\), and similarly for \(\upperSbound\) and \(\lowerSbound\).
Based on this property, we can show that the function \(\certificate\) decreases in expectation along the sample paths of the state process until it reaches the target.
If the value of \(\certificate\) is initially at most \(\lowerRAbound\), we have not left the sub-\(\upperRAbound\) set with some high probability (to be defined later) and thus have avoided the unsafe states.
Similarly, if within the target there exists a subset where the values of \(\certificate\) are sufficiently low, and they are expected to decrease outside of this set, the process will be ``trapped'' within this smaller subset with high probability (again, to be defined later).
These properties motivate the conditions of Definition~\ref{df:ras-C}, and are outlined in Figure~\ref{fig:sets-example}.
This intuitive explanation of reach-stay-avoid certification can be stated rigorously as the following theorem.

\begin{figure}[tb]
    \centering
    \includegraphics{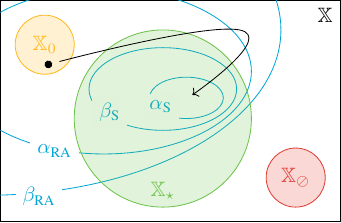}
    \caption{%
    A (time-homogenous) example of the sets of Definition~\ref{df:ras-C}.
    A sample trajectory \includegraphics{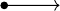} first reaches the target set \(\StatesTarget\) without staying, and then reaches it again and stays.
    Our goal is to find the latter moment (with high probability).%
    }
    \label{fig:sets-example}
\end{figure}

\begin{theorem}[RAS-C certifies an RAS specification]\label{th:main-result}
  A policy \(\policy\) satisfies an RAS specification \((\StatesTargetFamily, \StatesUnsafeFamily, \StatesInitial, \probRA, \probS)\) if there exists a reach-avoid-stay certificate such that
  \begin{align*}
    \probRA &= 1 - \tfrac{\lowerRAbound}{\upperRAbound}
    &\text{and}&
    &\probS &= 1 - \tfrac{\lowerSbound}{\upperSbound}.
  \end{align*}
\end{theorem}

\begin{remark}
  Note that there are two probabilities \(\probRA\) and \(\probS\) in the definition of the RAS specification, but four numbers in the definition of RAS-C.
  Since the probabilities of an RAS specification are defined by the ratios, we can fix either \(\lowerRAbound\) or \(\upperRAbound\) beforehand, and similarly for \(\lowerSbound\) or \(\upperSbound\).
\end{remark}

We prove Theorem~\ref{th:main-result} in the technical appendix available in the extended version of the paper. We present a sketch of the proof here.

\begin{proof}[Proof sketch for Theorem~\ref{th:main-result}]
First, we construct a random variable representing the first time the value of RAS-C for a given sample path leaves the interval \([\lowerSbound, \upperRAbound)\):
\begin{equation}\label{eq:the-interval-stopping-time}
  \define{\theStoppingTime}{\textstyle{\inf_{t \geq 0}} \bigl\{t \bigm| \certificate(t, \trajectoryProcess{\policy}{0}{x_0}{t}) \notin [\lowerSbound, \upperRAbound) \bigr\}.}
\end{equation}
Next, we show that it is almost surely a finite stopping time adapted to the state process.
Then, we construct the following stopped stochastic process that represents the evolution of the certificate and show that it is a supermartingale.

\begin{definition}
  Given the natural filtration \(\NaturalFiltration=(\NaturalAlgebra_t)_{t\geq 0}\) of a stochastic state process \((\trajectoryProcess{\policy}{0}{x_0}{t})_{t \geq 0}\), and a \(\NaturalFiltration\)-stopping time \(\tau\), a \emph{\(\tau\)-stopped reach-stay-avoid supermartingale (\(\tau\)-RAS-SM)} is a stochastic process \((Y_t)_{t\geq 0}\) defined as
  \[\define{Y_t}{\certificate(t \wedge \tau, \trajectoryProcess{\policy}{0}{x_0}{t \wedge \tau}).}\]
\end{definition}

\begin{lemma}[the proof can be found in the technical appendix]\label{lm:rasm-is-supermartingale}
  \(\theStoppingTime\)-RAS-SM \((Y_t)_{t\geq 0}\) is an \(\NaturalFiltration\)-supermartingale.
\end{lemma}

Finally, we use maximal inequalities for supermartingales to bound the probabilities that when the chosen moment of time is reached, both of the events used in Definition~\ref{df:reach-avoid-specification} are guaranteed to happen with the required probabilities.
\end{proof}

Moreover, it follows from the proof of the theorem that our approach can be used to certify either of the reach-avoid or stay properties by themselves due to the following corollary (proof outlined in the technical appendix).

\begin{corollary}\label{cor:ra-s-separately}
  By omitting conditions \ref{cnd:goal} and \ref{cnd:stay}, we can use an RAS-C to certify reach-avoidance requirement \eqref{eq:reach-avoid} without the staying property.
  Similarly, the stay property \eqref{eq:stay} by itself can be certified by omitting conditions \ref{cnd:init} and \ref{cnd:safe}.
\end{corollary}

\section{Neural Certificate Training \& Verification}

Finding certificates analytically is a challenging task. Therefore, drawing on advancements in the new field of neural certificates, we employ a deep neural network to find the certificate \(\certificate_\theta\) as a function of an input state \(x\) and trainable network parameters \(\theta\).
Our method is summarized in Algorithm~\ref{alg:training} and Figure~\ref{fig:architecture}, and explained in this section.

\begin{algorithm}[htb]
\caption{RAS Certificate Training\label{alg:training}}
\SetKw{Continue}{continue}
\KwData{%
a drift \(f\) and diffusion \(g\) of \eqref{eq:general_sde}
, a policy \(\policy\), an RAS specification \((\StatesTargetFamily, \StatesUnsafeFamily, \StatesInitial, \probRA, \probS)\), batch size \(n\), regularization multiplier \(\lambda\), number of cells per dimension \(m\), maximum cell splitting depth \(k\), frequency of verification \(q\).}
\KwResult{%
satisfaction of the specification (\texttt{Yes}/\texttt{No}).
}
Initialize the levels \(\upperRAbound > \lowerRAbound > \upperSbound > \lowerSbound > 0\)\;
create AABB cells \(\Cells\) for the verifier\;

\For{\(N\) epochs}{%
  sample a batch \(\batch\) of \(n\) points from \(\States\)\;
  find the loss \(L\) using \eqref{eq:loss-total}--\eqref{eq:loss-regularizer}\;
  train the parameters \(\theta\) using the gradient \(\gradient{\theta} L\)\;
\If{every \(q\) steps}{%
using IBP, compute lower \(\certificate_{\text{low}}(C)\) and upper \(\certificate_{\text{up}}(C)\) bounds on \(\certificate(x)\) for all cells \(C\)\;
compute \(\estProbRA \) via \eqref{eq:emp-ra-prob} and continue if \(\estProbRA < \probRA\)\;
compute \(\estProbS\) via \eqref{eq:emp-s-prob} and continue if \(\estProbS < \probS\)\;
using IBP, compute the upper bounds \([\generator_\policy \certificate]_{\text{up}}(C)\) for cells in \(\Cells_2\) defined by \eqref{eq:C_2}\;
set the verification depth \(d \gets 0\)\;
find unverified cells \(\UnverifiedCells \gets \{C \mid [\generator_\policy \certificate]_{\text{up}}(C) \geq 0\}\)\;
\While{\(\lvert \UnverifiedCells \rvert > 0\) and \(d < k\)}
{
split each cell in \(\UnverifiedCells\) into \(2^l\) smaller cells\;
increase the depth \(d \gets d + 1\)\;
compute the upper bounds \([\generator_\policy \certificate]_{\text{up}}(C)\)\;
remove cells where the bound is negative\;
}
\If{\(\estProbRA \geq \probRA\) and \(\estProbS \geq \probS\) and \(\lvert \UnverifiedCells \rvert = 0\)}{%
\Return{\texttt{Yes}}\;
}
}
}
\Return{\texttt{No}.}
\end{algorithm}

\begin{figure}[tb]
  \centering
  \includegraphics{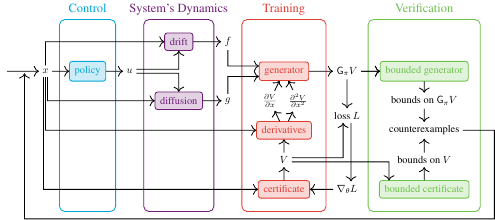}
  \caption{%
    Architecture of the neural certificate training and verification system.%
  }
  \label{fig:architecture}
\end{figure}

While the theoretical contributions separate the time \(t \in \NonNegativeReals\) and the state \(x \in \States\) variables due to different requirements on differentiability of the certificate with respect to them.
When this distinction is of no concern, we can embed the time variable \(t\) into the space vector (for example, as \(x = (e^{-t}, x_1, \dots, x_l)\) to map it into a compact interval).
Moreover, if the system is time-homogeneous, we can drop the dependence of the certificate on time entirely.
Technically, this makes the initial condition of  Definition~\ref{df:ras-C} more restrictive, as it will require that \(\certificate(0, x) \leq \lowerRAbound\) for all \((t,x) \in \NonNegativeReals \times \StatesInitial\) instead of just for \(t = 0\), but the resulting dimensionality reduction simplifies the verification process.
For these reasons, we opt to make the presentation more succinct by dropping the time variable from now on.

\subsection{Training}

This approach raises a few issues to be addressed.

First, we need to set the threshold levels.
We elect:
\begin{align*}
  \lowerRAbound &= 1, &\upperRAbound &= \tfrac{\kappa\lowerRAbound}{1 - \probRA},&
  \upperSbound &= 0.9, & \lowerSbound &= \tfrac{\upperSbound(1 - \probS)}{\kappa}.
\end{align*}
The constant \(\kappa > 1\) is chosen to force the training to find larger probabilities \(\probRA\) and \(\probS\) in order to give some slack to the verification procedure.

Next, to ensure that the certificate is twice continuously differentiable, we employ smooth activation functions such as the hyperbolic tangent and soft-plus.
To ensure that the nonnegativity condition of Definition~\ref{df:ras-C} holds, we add a non-negative transformation to the last layer.

Then, armed with Theorem~\ref{th:main-result}, we train the network as follows.
To ensure that the conditions of Definition~\ref{df:ras-C} hold, we sample batches of points from the various sets used therein.
We then check how much the points violate the conditions, and sum up the violations into the loss function.

For the initial, safety, and goal conditions, this requires evaluating the certificate which can be done with a simple forward pass through the certificate network.

We combine the decrease and stay conditions by setting the thresholds \(\generatorBoundRA(t)\) and \(\generatorBoundS(t)\) to the same constant \(\generatorBoundRA\).
Then we check both conditions over the union of their respective sets.
This requires evaluating the generator \(\generator_\policy \certificate(x)\) at each sampled point \(x\) as defined by \eqref{eq:generator}.
First, we obtain the control signal \(u = \policy(x)\) from the policy network (or a function), and compute the drift \(\drift_\policy(x, u)\) and diffusion \(\diffusion_\policy(x, u)\) of the SDE \eqref{eq:general_sde}.
Next, we find the first and second derivatives of the certificate network analytically using the method of \citet{singla2020secondorder}.
The details are provided in the technical appendix.
Since evaluation of the derivatives is costly and only required for some of the points, we separate the certificate and derivative networks.

Finally, we regularize the loss by the product of operator norms of the neural network's weight matrices \(\lVert \placeholder \rVert_{\infty\to\infty}\) defined as the maximum \(\ell_1\)-norm of the rows of a matrix.
This ensures that the weights of the certificate network remain small, making the verification task easier.

The resulting loss \(L\) is equal to
\begingroup
\allowdisplaybreaks
\begin{align}
  \define{L&}{L_{\oslash} + L_0 + L_{\star} + L_{\downarrow} + R,}\qquad\text{where}\label{eq:loss-total}\\
  \define{L_{\oslash}&}{\textstyle{\sum_{x \in \batch \cap \StatesUnsafe}}\bigl(\upperRAbound - \certificate_{\theta}(x)\bigr)^+,}\label{eq:loss-unsafe}\\
  \define{L_0&}{\textstyle{\sum_{x_0 \in \batch \cap \StatesInitial}}\bigl(\certificate_{\theta}(x_0) - \lowerRAbound\bigr)^+,}\label{eq:loss-initial}\\
  \define{L_{\star}&}{\textstyle{\sum_{x \in \batch \cap \StatesTarget}}\bigl(\upperSbound - \certificate_{\theta}(x)\bigr)^+,}\label{eq:loss-out}\\
  \define{L_{\downarrow}&}{\textstyle{\sum_{x \in \batch \cap (L^-_{\upperRAbound} \setminus L^-_{\lowerSbound})}}\bigl(\generator_\policy \certificate_{\theta}(x) + \generatorBoundRA\bigr)^+,}\label{eq:loss-decrease}\\
  \define{R&}{\lambda \textstyle{\prod_{i=0}^{N-1}}\lVert \mathbf{W}^{(i)} \rVert_{\infty\to\infty}.}\label{eq:loss-regularizer}
\end{align}
\endgroup

\subsection{Verification}
To verify that all of the properties of Definition~\ref{df:ras-C} hold, we employ interval bound propagation (IBP) \cite{xu2020automatic}.
IBP, as well as other abstract interpretation methods \cite{DBLP:conf/sp/GehrMDTCV18,kouvaros2021towards,zhang2018efficient}, have proven effective in verifying correctness of neural networks' output. 
To this end, we construct bounded certificate and generator networks that take intervals rather than point values as their inputs.
By propagating the input intervals through these networks, we obtain upper and/or lower bounds on the certificate and its generator.

To construct the interval inputs for IBP, we divide the set \(\States\) into cells \(\Cells\) which are axis-aligned bounding boxes (AABBs) covering the whole space \(\States\), that is, \(\States \subseteq \bigcup_{C \in \Cells} C\).
We do this by discretizing each dimension \([x_i^{\text{min}}, x_i^{\text{max}}]\), \(i=1, \dots, l\) of the state space \(\States\) into collections \(\mathcal{X}_i\) of \(m\) equal intervals, \(\define{\mathcal{X}_i}{\{[x_i^{\text{min}}, x_i^{\text{min}}+\Delta_i], \dots, [x_i^{\text{max}}-\Delta_i, x_i^{\text{max}}]\}}\), \(\define{\Delta_i}{(x_i^{\text{max}} - x_i^{\text{min}})/l}\), and using all of the possible Cartesian products thereof as \(m^l\) cells, \(\define{\Cells}{\mathcal{X}_1 \times \cdots \times \mathcal{X}_l}\).

For each cell, we compute the lower \(\certificate_{\text{low}}(C)\) and upper \(\certificate_{\text{up}}(C)\) bounds on \(\certificate(x)\) within that cell, \(x \in C\).
Given the bounds, we find the empirical reach-avoid probability \(\estProbRA\)
\begin{align}
  \define{\estProbRA &}{(1 - \estLowerRAbound / \estUpperRAbound)^+,}\label{eq:emp-ra-prob}&
  \text{where}\\
  \define{\estLowerRAbound &}{\max_{C : C \cap \StatesInitialFamily \neq \emptyset} \certificate_{\text{up}}(C),}&
  \define{\estUpperRAbound &}{\min_{C : C \cap \StatesUnsafeFamily \neq \emptyset} \certificate_{\text{low}}(C).}\nonumber
\end{align}
This verifies that the initial set values are at most \(\estLowerRAbound\), and the unsafe set values are at least \(\estUpperRAbound\).
Therefore, if the decrease condition is satisfied, the reach-avoid event \eqref{eq:reach-avoid} occurs with probability at least \(\estProbRA\).

Similarly, for the goal condition, we compute a provable stay probability \(\estProbS\) as
\begin{align}
  \define{\estProbS &}{(1 - \estLowerSbound / \estUpperSbound)^+,}\label{eq:emp-s-prob}&
  \text{where}\\
  \define{\estLowerSbound &}{\min_{C : C \cap \StatesTargetFamily \neq \emptyset} \certificate_{\text{up}}(C),}&
  \define{\estUpperSbound &}{\min_{C : C \cap \partial\StatesTargetFamily \neq \emptyset} \certificate_{\text{low}}(C).}\nonumber
\end{align}
This guarantees that at least in one of the cells within the target the certificate values go below \(\estLowerSbound\), and on the boundary the values are at least \(\estUpperSbound\).

Finally, we verify the decrease condition by computing upper bounds on the generator values \([\generator_\policy \certificate]_{\text{up}}(C)\) in the set
\begin{equation}\label{eq:C_2}
  \define{\Cells_2}{\bigl\{C \bigm| \certificate_{\text{low}}(C) > \estLowerSbound \wedge \certificate_{\text{up}}(C) \leq \estUpperRAbound\bigr\}}.
\end{equation}
We then check if any of the cells have non-negative upper bounds.
If they do (and we observed this happening often, as the generator bounds involve quadratic terms and as a result they are much looser than the value bounds) the decrease condition in those cells cannot be verified.
We split such cells in two along each of the dimensions. We then recompute the cell bounds for each of the \(2^l\) sub-cells.
We repeat this until either no counterexample is found, meaning that the decrease condition is verified (as it holds over a collection of cells fully covering the desired subspace), or the maximum iteration depth is achieved. 

Finally, if the estimated probabilities \(\estProbRA\) and \(\estProbS\) exceed the specified ones \(\probRA\) and \(\probS\), and no cells with \(\generator_\policy \certificate_{\text{up}}(C) \geq 0\) are found, we obtained a provably correct RAS-C.

\section{Experiments}\label{sec:experiments}

To illustrate the practical efficacy of our framework, we synthesize neural certificates for two continuous-time control benchmarks.
To the extent of our knowledge, this is the first work on formal verification of neural certificates for continuous-time stochastic control using neural networks.
Thus, we do not include comparisons to any other methods.

The computing infrastructure and values of the hyperparameters for all of the experiments and the neural network architectures are given in the technical appendix.

We employ \texttt{auto\_LiRPA} \cite{xu2020automatic} for IBP, which is a library for automatic bound computation with linear relaxation based perturbation analysis.
We integrate \eqref{eq:general_sde} numerically with stochastic Runge--Kutta method using \texttt{torchsde} \cite{li2020scalable}.
Note that this approximation scheme is not employed in Algorithm~\ref{alg:training}, but only to simulate the sample paths shown in the figures.

We successfully train RAS certificates for two controlled SDEs.
The times and numbers of verification phases required for training are summarized in Table~\ref{tbl:times}.
More details on each of the problems are given in the technical appendix.

\begin{table}[htb]
  \centering
  \begin{tabular}{lrrrr}
    \toprule
    & \multicolumn{2}{c}{Time per} & \multicolumn{2}{c}{Number of} \\
    & \multicolumn{2}{c}{verification} & \multicolumn{2}{c}{verifications}\\
    \cmidrule(lr){2-3} \cmidrule(l){4-5}{SDE} & mean & st. dev. & mean & st. dev.\\
    \midrule
    Inverted pendulum & 14.9 s & 6.4 s & 3.2 & 0.8 \\
    Bivariate GBM &  21.1 s & 14.4 s & 11.2 & 4.0 \\
    \bottomrule
  \end{tabular}
  \caption{%
    Runtimes for the experiments.
    Each data point is based on five randomly generated seeds.
    The seeds and data for each run are included in the code appendix.%
  }
  \label{tbl:times}
\end{table}

\subsection{Stochastic Inverted Pendulum}

Inverted pendulum \cite{chang2019neural,wu2023neural} is a classical nonlinear control problem.
We disturb the angle by a Wiener process with scale \(\sigma\), resulting in the following SDEs:
\begin{align*}
    \dd \pndVelocity_t &=  \bigl(\tfrac{\pndGravity}{\pndLength}  \sin \pndAngle_t + \tfrac{\pndMaxTorque u_t - \pndFriction \pndVelocity_t}{\pndMass \pndLength^2}\bigr)\,\dd t + \sigma\,\dd\WienerProcess_t,&
    \dd \pndAngle_t &= \pndVelocity_t \,\dd t,
\end{align*}
where \(\pndVelocity\) is the angular velocity.
The drift and diffusion are
\begin{align*}
\drift_\policy(\pndVelocity,\pndAngle) &= \begin{bmatrix}
\frac{\pndGravity}{\pndLength}  \sin \pndAngle + \frac{\pndMaxTorque \policy(\pndVelocity,\pndAngle) - \pndFriction \pndVelocity}{\pndMass \pndLength^2} \\
\pndVelocity
\end{bmatrix}, &
\diffusion_\policy(\pndVelocity,\pndAngle) &= \begin{bmatrix}
\sigma \\
0
\end{bmatrix}.
\end{align*}
The policy \(\policy(\pndVelocity,\pndAngle)\) is represented by a neural network.
An example of a trained certificate is presented in Figure~\ref{fig:certificate-example}.

\subsection{Bivariate Geometric Brownian Motion}

Geometric Brownian motion (GBM) is one of the most studied stochastic processes and is frequently utilized in financial mathematics.
We consider the following controlled version of bivariate GBM:
\begin{align*}
  \dd \stateProcess_t &= \bigl(\mu \stateProcess_t - u_t(\stateProcess_t)\bigr)\,\dd t + \sigma(\stateProcess_t)\,\dd \WienerProcess_{t}, &\text{where}\\
  \mu &= \begin{bmatrix}
    -0.5 & 1 \\
    -1 & -0.5 
  \end{bmatrix} \quad\text{and}\quad \sigma=0.2 \diag(\stateProcess_t).
\end{align*}
We use a stabilizing policy \(\policy(x) = -x\) which pushes the system towards the equilibrium point.
The resulting certificate is visualized in Figure~\ref{fig:certificate-bvgbm}.

\begin{figure}[htb]
  \centering
  \includegraphics[width=\columnwidth]{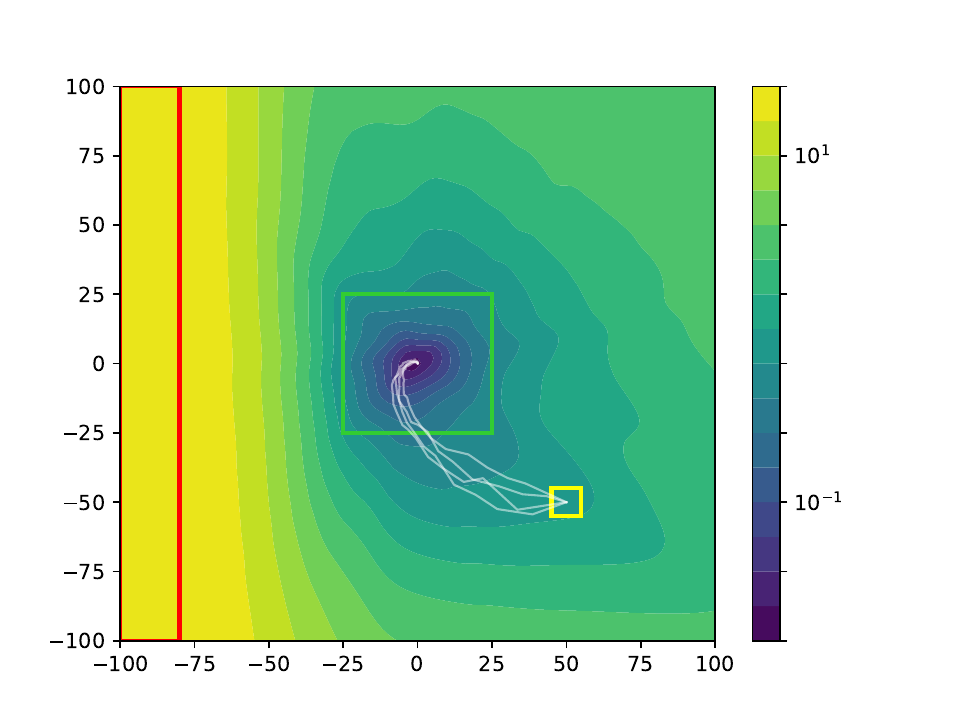}
  \caption{%
    An example of a neural supermartingale certificate for the bivariate GBM problem.%
  }
  \label{fig:certificate-bvgbm}
\end{figure}

\section{Related Work}\label{sec:related}

The concept of representing Lyapunov functions as neural networks was theoretically discussed in seminal work~\cite{long1993feedback,prokhorov1994lyapunov}. This idea led to the development of numerical machine learning algorithms for neural Lyapunov functions~\cite{serpen2005empirical,DBLP:conf/ijcnn/PetridisP06,noroozi2008generation}. 
Building on this, learning algorithms have been extended to encompass barrier functions for avoidance control~\cite{richards2018lyapunov,dawson2022safe}, contraction metrics~\cite{DBLP:conf/corl/SunJF20}, and the compositional certification of multi-agent systems~\cite{DBLP:conf/iclr/QinZCCF21,DBLP:conf/l4dc/ZhangXQF23}, and the simultaneous training of control policies 
and neural certificates~\cite{kolter2019learning,wu2023neural}, 
albeit initially without formal soundness guarantees. 
Our result provides formal soundness guarantees, drawing upon and extending the following related work.

\subsubsection{Neural Continuous-Time Lyapunov Certificates.}
Coupling machine learning with symbolic reasoning techniques like satisfiability modulo theory (SMT) solving and bound propagation made a significant step forward for neural certification with formal guarantees~\cite{chang2019neural,abate2020formal}.
Initially focused on Lagrangian and asymptotic stability, neural certificates with provable guarantees for continuous-time deterministic systems have since been extended to cover avoidance, reachability, invariance, persistence, and combinations of these properties~\cite{DBLP:conf/hybrid/ZhaoZC020,takeishi2021learning,DBLP:conf/nips/Zhang0V023,edwards4880686general}. 

\subsubsection{Neural Discrete-Time Lyapunov Certificates.} 
Advances in neural certificates for differential equations have inspired the development of their discrete-time counterparts. Initially applied to stability and reachability analysis of difference equations and hybrid systems~\cite{DBLP:conf/hybrid/ChenFMPP21}, neural discrete-time Lyapunov certificates have been extended to control under avoidance and persistence specifications~\cite{ANAND20232431,chen2021learning}, and have enabled symbolic reasoning with efficient neural network verifiers~\cite{yanglyapunov,mandal2024formally}. As reasoning about discrete-time systems is akin to reasoning about programs, neural certificates have also been applied to termination analysis~\cite{DBLP:conf/sigsoft/GiacobbeKP22}.

\subsubsection{Neural Discrete-Time Supermartingale Certificates.}
In the presence of stochastic uncertainty, traditional Lyapunov-like certificates are overly conservative and unrealistic; relying on worst-case reasoning, they do not quantify the probability of events to occur. Probabilistic neural certificates, which build on supermartingale-like conditions, were initially developed for almost-sure termination and, analogously, discrete-time reachability~\cite{abate2021learning,ansaripour2023learning}. These neural supermartingale certificates were later generalized to 
quantitative reachability, avoidance and reach-avoidance~\cite{DBLP:journals/csysl/MathiesenCL23,vzikelic2023learning,DBLP:conf/tacas/ChatterjeeHLZ23,badings2024learning}. Further advancements have studied the optimization of probability bounds and compositional certification~\cite{DBLP:conf/concur/AbateEGPR23,vzikelic2024compositional}.

\subsubsection{Neural Continuous-Time Supermartingale Certificates.} 
Neural supermartingale-based techniques have been studied primarily in the context of discrete-time systems. Their extension to continuous-time systems has been solely explored in the context of stability control, and even then, without provable guarantees~\cite{zhang2022neural}.

\section{Conclusion \& Future Work}\label{conclusions}

We have introduced the first neural supermartingale for continuous-time reasoning with provable guarantees, applicable to combinations of reachability, avoidance, and persistence properties. We have formulated a proof rule (Theorem~\ref{th:main-result}) for a supermartingale certificate for continuous-time systems and, by integrating machine learning with symbolic reasoning, we have fully automated their construction. We have built a prototype and demonstrated that our method is practically effective on continuous-time systems with neural policies, extending the state of the art in algorithmic verification for provably safe stochastic nonlinear control.

Our framework is open to extension towards the automated synthesis of neural controllers alongside their certificates, without the need to pre-initialize a control policy or guide it with other methods. We foresee integration with advanced symbolic reasoning techniques based on adaptive discretization~\cite{vzikelic2023learning,badings2024learning}, adversarial attack algorithms~\cite{yanglyapunov,wu2023neural}, and generalization to more expressive temporal logic requirements~\cite{DBLP:conf/cav/AbateGR24,DBLP:conf/aaai/NadaliM0024,neuralmc}.

\section*{Technical Appendix}\label{app:a}

This appendix contains a simple example of the problem at hand, the proofs of the lemmas presented in the paper, and implementation details for the experimental section.

\subsection{An Illustrative Example}

Consider the following stochastic differential equation:
\begin{equation}\label{eq:toy-example}
  \dd \stateProcess_t = a \stateProcess_t\,\dd t + \policy(\stateProcess_t)\,\dd\WienerProcess_t.
\end{equation}
When \(u \equiv 0\) there is no stochasticity to the problem, and the resulting deterministic system is unstable for any starting point \(x_0\), since the analytical solution \(x(t) = x_0 e^{at}\) to \(\dot x = a x\) with the initial condition \(x(0) = x_0\) tends to infinity.

Now consider the following control policy:
\begin{equation}\label{eq:toy-policy}
  \policy_\sigma(x) \equiv \sigma x.
\end{equation}
In other words, the controller adds a white noise proportional to the state.
The resulting system's dynamics are
\[
  \dd \stateProcess_t = a \stateProcess_t\,\dd t +  \sigma \stateProcess_t\,\dd\WienerProcess_t
\]
which is a geometric Brownian motion with drift \(a\) and volatility \(\sigma\).
It is a well-studied stochastic process and the solution for the initial condition \(\stateProcess_0 = x_0\) is
\[
  \stateProcess_t = x_0 e^{\bigl(a - \frac{\sigma^2}{2}\bigr)t + \sigma \WienerProcess_t}.
\]
If \(\sigma > \sqrt{2a}\), surprisingly (and even somewhat counterintuitively), the system becomes stabilized by pure noise injection.
This highlights the fundamentally different nature of continuous-time stochastic problems compared to their deterministic counterparts.
This phenomenon is illustrated by Figure~\ref{fig:toy-example}.

\begin{figure}[htb]
  \centering
  \includegraphics[width=\columnwidth]{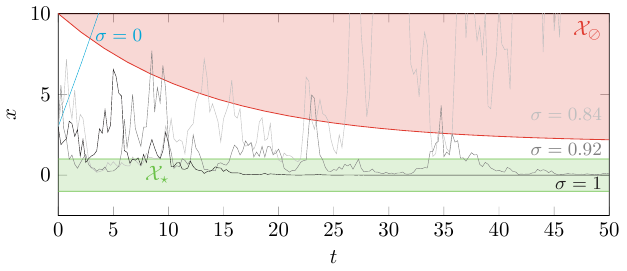}
  \caption{%
    Sample paths of three geometric Brownian motions solving \eqref{eq:toy-example} with \(a=0.4\) and policies \(\policy_\sigma\) given by \eqref{eq:toy-policy} when \(\sigma\) takes values of 1, 0.92, and 0.84, and 0.%
  }
  \label{fig:toy-example}
\end{figure}

Now consider three control policies with \(\sigma \in \{1.0, 0.92, 0.84\}\).
Let us restric our attention to the set \(\States = [-1,10]\).
Suppose we are given the following time-heterogenous unsafe states and time-homogenous targets:
\begin{align*}
  \StatesUnsafeFamily &= \bigl([2 + 8e^{-0.075t}, 10]\bigr)_{t\geq 0}
  &\text{and}&
  &\StatesTargetFamily &=\bigl([-1,1]\bigr)_{t\geq 0}.
\end{align*}

Figure~\ref{fig:toy-example} illustrates this problem and shows sample paths of the three systems.
The sample path under the policy \(\policy_{1}\) satisfies the reach-avoid-stay property at time 10;
the sample path under the policy \(\policy_{0.84}\) does not satisfy the stay property, but it does satisfy the reach-avoidability.
More interestingly, the policy \(\policy_{0.92}\) is stable (converges to the equilibrium) so it satisfies the stay property, but not the reach-avoid one.

This example inspects a single sample path for each process.
In practice, for each of the processes there exist (a.s.)  paths failing to satisfy the reach-avoid-stay criteria, especially when they start close to the unsafe set.
This is why we aim to verify RAS satisfaction probabilistically and restrict ourselves to some initial set.

The generator of a system driven by \eqref{eq:toy-example}--\eqref{eq:toy-policy} is
\[
  \generator = \frac{\partial}{\partial t}
  + a x \frac{\partial}{\partial x}
  + \frac{1}{2}(\sigma x)^2 \tfrac{\partial^2}{\partial x^2}.
\]

\subsection{Analytical Derivatives of a Neural Network}

We compute \(\tfrac{\partial \certificate}{\partial x}\) and \(\tfrac{\partial^2 \certificate}{\partial x^2}\) using the following statement.

\begin{proposition}[\citet{singla2020secondorder}, Lemma~1]\label{pr:hessian}
  Consider an \(N\)-layer neural network defined recursively for \(i=0,1,\dots,N-1\) as
  \[
    \mathbf{a}^{(-1)} = \mathbf{x},\quad
    \mathbf{z}^{(i)} = \mathbf{W}^{(i)} \mathbf{a}^{(i-1)} + \mathbf{b}^{(i)},\quad
    \mathbf{a}^{(i)} = \sigma_i(\mathbf{z}^{(i)}).
  \]
  The \(j\)-th row of the Hessian of its output \(\mathbf{z}^{(N-1)}\) with respect to the input \(\mathbf{x}\) can be computed via
  \begin{align*}
    \Hessian{\mathbf{x}}\mathbf{z}_j^{(N-1)} &= \sum_{i=0}^{N-2} (\mathbf{B}^{(i)})^\top \diag \bigl(\mathbf{F}_j^{(N-1, i)} \odot \sigma_i''(\mathbf{z}^{(i)})\bigr) \mathbf{B}^{(i)};\\
    \mathbf{B}^{(i)} & = \begin{cases}
      \mathbf{W}^{(0)}, &i = 0,\\
      \mathbf{W}^{(i)} \diag\bigl(\sigma_{i-1}'(\mathbf{z}^{(i-1)})\bigr)\mathbf{B}^{(i-1)}, &\i \geq 1;
    \end{cases}\\
    \mathbf{F}^{(k,i)} & = \begin{cases}
      \mathbf{W}^{(k)}, &i = k-1,\\
      \mathbf{W}^{(k)} \diag\bigl(\sigma_{k-1}'(\mathbf{z}^{(k-1)})\bigr)\\
      \quad \cdot \mathbf{F}^{(k-1,i)}, &i \leq k-2,
      \end{cases}
  \end{align*}
  and its Jacobian is equal to \(\mathbf{B}^{(N-1)}\).
\end{proposition}

\begin{remark}
  Our neural network architecture is slightly different, with the final output \(\mathbf{a}^{(N)} = \sigma_N(\mathbf{z}^{(N)})\) instead of \(\mathbf{z}^{(N)}\).
  We obtain the formulae for our case by extending the network of Proposition~\ref{pr:hessian} with a final linear layer with \(\mathbf{W}^{(N+1)} = [1]\) an \(\mathbf{b}^{(N+1)} = 0\) in the calculations.
\end{remark}

Given this proposition, we find the first derivative vector as the transpose of the Jacobian, and the second derivative vector as the diagonal of the Hessian.

\subsection{Proof of Theorem~\ref{th:main-result}}

First, we need to ensure that \(\theStoppingTime\) is indeed a stopping time.

\begin{lemma}\label{lm:finite-stopping-time-RAS}
  Let \(\NaturalFiltration = (\NaturalAlgebra_t)_{t\geq 0}\) be the natural filtration with respect to a state process \((\trajectoryProcess{\policy}{0}{x_0}{t})_{t\geq 0}\) issuing in some state \(x_0 \in \StatesInitial\).
  Consider a continuous function \(\certificate(t, x) : \NonNegativeReals \times \States \to \Reals\).
  For any constants \(\lowerSbound,\upperRAbound \in \Reals\) such that \(\beta \leq \rho\), the random variable \(\theStoppingTime\) given by \eqref{eq:the-interval-stopping-time} is an \(\NaturalFiltration\)-stopping time.
  Moreover, if the decrease and stay conditions of Definition~\ref{df:ras-C} hold, then \(\theStoppingTime < \infty\) (a.s.).
\end{lemma}

In proving Lemma~\ref{lm:finite-stopping-time-RAS}, we will employ the following result.

\begin{proposition}[recurency citerion, cf. \citet{khasminskii2011stability}, Theorem 3.9]\label{pr:recurrence}
 A Feller--Dynkin process \((\eta_t)_{t\geq 0}\) with infinitesimal generator \(\generator\) leaves a domain \(\mathbb{U}\) in finite time (a.s.) if it is regular (i.e., defined a.s. for all \(t\geq 0\)) and there exists in \(\NonNegativeReals \times \mathbb{U}\) a non-negative function \(\certificate(t, x)\), twice continuously differentiable with respect to \(x\) and continuously differentiable with respect to \(t\), such that \(\generator \certificate(t,x) \leq -\varsigma(t)\) for some non-negative function \(\varsigma(t)\) satisfying
  \[\lim_{t \to \infty}\int_0^t \varsigma(s)\,\dd s = \infty.\]
\end{proposition}

\begin{remark}
  Note that Proposition~\ref{pr:recurrence} requires the process to be regular which under our assumptions follows from continuity \cite[p.~75]{khasminskii2011stability}.
\end{remark}

\begin{proof}[Proof of Lemma~\ref{lm:finite-stopping-time-RAS}]
  Note that because \(\certificate\) is continuous and therefore preserves Borel-measurability, both of the events \(\{\certificate(t, \trajectoryProcess{\policy}{0}{x_0}{t}) < \lowerSbound\}\) and \(\{\certificate(t, \trajectoryProcess{\policy}{0}{x_0}{t}) \geq \upperRAbound\bigr\}\) are Borel-measurable;
  therefore, their union is also measurable.
  Thus, \(\theStoppingTime\) is the first hit time of a measurable set, and by the d\'ebut theorem it is a stopping time.
  Consider next a domain \(\mathbb{U} = \bigl\{x \bigm| \lowerSbound \leq \certificate(t, x) < \upperRAbound\bigr\}\).
  The process \((\trajectoryProcess{\policy}{0}{x_0}{t})_{t \geq 0}\) has a generator \(\generator_\pi\) which satisfies on \(\mathbb{U}\) the condition \(\generator_\policy \certificate(t,x) \leq -\varsigma(t)\) for \(\varsigma(t) = \generatorBoundRA(t) \vee \generatorBoundS(t)\) due to the decrease and stay conditions.
  Since
  \[\int_0^t \varsigma(s)\,\dd s = \int_0^t \generatorBoundRA(s) \vee \generatorBoundS(s)\,\dd s \geq \int_0^t \generatorBoundRA(s)\,\dd s,\]
  the conditions of Proposition~\ref{pr:recurrence} are satisfied and the stopping time \(\theStoppingTime\) is finite (a.s.).
\end{proof}

Next, we use the stopping time \(\theStoppingTime\) to construct the following supermartingale.

Non-negative supermartingales such as \(\theStoppingTime\)-RAS-SM have the following properties useful to us.

\begin{proposition}[optional stopping theorem, cf. \citet{gall2016brownian}, Theorem 3.25]\label{pr:optional-stopping}
  Let \((\eta_t)_{t\geq 0}\) be a non-negative supermartingale with right-continuous sample paths with respect to some filtration \((\Algebra)_{t\geq 0}\).
  Let \(\tau_1\) and \(\tau_2\) be two stopping times such that \(\tau_1 \leq \tau_2\).
  Then, \(\eta_{\tau_1}\) and \(\eta_{\tau_2}\) are in \(L_1\) and
  \[
    \eta_{\tau_1} \geq \expect [\eta_{\tau_2}\mid \Algebra_{\tau_1}].
  \]
\end{proposition}

\begin{corollary} For such a process \((\eta)_{t\geq 0}\),
  \(\expect[\eta_{\tau_1}] \geq \expect [\eta_{\tau_2}].\)
  \label{cr:decreasing-expectation}
\end{corollary}

\begin{proof}
  Follows immediately by applying the law of total expectation to the statement of Proposition~\ref{pr:optional-stopping}.
\end{proof}

\begin{proposition}[Chebyshev's inequality, cf. \citet{stein2009real}, p.~91]\label{pr:Markov}
  If \(f\) is a non-negative function, \((\eta_t)_{t\geq 0}\) is a stochastic process such that \(\expect\bigl[f(\eta_t)\bigr]\) exists, and \(r > 0\), then 
  \[
    \probability \bigl\{f(\eta_t) \geq r \bigr\} \leq \frac{\expect \bigl[ f(\eta_t) \bigr]}{r}.
  \]
\end{proposition}

\begin{proposition}[Maximal inequality for non-negative supermartingale, cf. \citet{prajna2004stochastic}, Lemma 6]\label{lm:maximal}
  Given a filtration \(\filtration = (\Algebra)_{t\geq 0}\), let \((\eta_t)_{t\geq 0}\)  be a non-negative \(\filtration\)-supermartingale with right-continuous sample paths.
  Then for every \(r > 0\),
  \[
    \probability \Bigl\{\sup_{t \geq 0} \eta_t \geq r \Bigr\} \leq \frac{\expect[\eta_0]}{r}.
  \]
\end{proposition}

Finally, armed with all the necessary properties of an RAS-SM, we are ready to prove Theorem~\ref{th:main-result}.

\begin{proof}[Proof of Theorem~\ref{th:main-result}]
  Consider the supermartingale \((Y_t)_{t\geq 0}\) of Lemma~\ref{lm:rasm-is-supermartingale}.
  By its construction and the non-negativity condition for RASMs, it is non-negative, and by the assumptions of Theorem~\ref{th:main-result} it is continuous, since the RAS-C \(\certificate\) is twice continuously differentiable with respect to \(x\) and continuously differentiable with respect to \(t\), both implying continuity.
  Therefore, \(\expect[Y_{\tau}] \leq \expect [Y_{0}]\) for any stopping time \(\tau\) by Corollary~\ref{cr:decreasing-expectation}.
  Note that 
  \begin{align}
    Y_{\tau} &= \certificate(\tau \wedge \tau, \trajectoryProcess{\policy}{0}{x_0}{\tau \wedge \tau}) = \certificate(\tau , \trajectoryProcess{\policy}{0}{x_0}\tau), &\text{and}\label{eq:L-tau}\\
    Y_{0} &= \certificate(0, \trajectoryProcess{\policy}{0}{x_0}{0}) = \certificate(0, x_0).\label{eq:L-0}
  \end{align} 
  The initial condition of Definition~\ref{df:ras-C} implies
  \begin{equation}
    \expect\bigl[\certificate(0, \trajectoryProcess{\policy}{0}{x_0}{0})\bigr] = \certificate(0,x_0) \leq \lowerRAbound.\label{eq:L-0-bound}
  \end{equation}
  Equations \eqref{eq:L-tau}--\eqref{eq:L-0-bound} imply 
  \[\expect\bigl[\certificate(\tau , \trajectoryProcess{\policy}{0}{x_0}{\tau})\bigr] \leq \expect\bigl[\certificate(0, \trajectoryProcess{\policy}{0}{x_0}{0})\bigr] = \certificate(0,x_0) \leq \lowerRAbound.\]
  By Proposition~\ref{pr:Markov}, this implies
  \[
    \probability \bigl\{\certificate(\tau, \trajectoryProcess{\policy}{0}{x_0}{\tau}) \geq \upperRAbound \bigr\} \leq \tfrac{1}{\upperRAbound}\expect\bigl[\certificate(\tau , \trajectoryProcess{\policy}{0}{x_0}{\tau})\bigr] \leq \tfrac{\lowerRAbound}{\upperRAbound}.
  \]
  
  Now consider the stopping time \(\theStoppingTime\) of Lemma~\ref{lm:finite-stopping-time-RAS}.
  By its construction, either
  \[
    \certificate(\theStoppingTime, \trajectoryProcess{\policy}{0}{x_0}{\theStoppingTime}) \leq \lowerSbound
    \quad\text{or}\quad
    \certificate(\theStoppingTime, \trajectoryProcess{\policy}{0}{x_0}{\theStoppingTime}) > \upperRAbound.
  \]
  Since \(\lowerSbound \leq \upperRAbound\), these events are incompatible, and thus
  \begin{align*}
    \probability\bigl\{\certificate(\theStoppingTime, \trajectoryProcess{\policy}{0}{x_0}{\theStoppingTime}) &\leq \lowerSbound\bigr\}\\
    &= 1 - \probability\bigl\{\certificate(\theStoppingTime, \trajectoryProcess{\policy}{0}{x_0}{\theStoppingTime}) \geq \upperRAbound\bigr\}\\
    &\geq 1 - \tfrac{\lowerRAbound}{\upperRAbound} = \probRA.
  \end{align*}
  Since \(\certificate(t, \trajectoryProcess{\policy}{0}{x_0}{t}) \leq \upperRAbound\) for all \(t < \theStoppingTime\) by the definition of the stopping time \(\theStoppingTime\), it follows from the safety condition that \(\trajectoryProcess{\policy}{0}{x_0}{t} \neq \StatesUnsafe\) for all \(t < \theStoppingTime\).
  Thus,
  \[
    \probability\Bigl\{\bigl(\trajectoryProcess{\policy}{0}{x_0}{\theStoppingTime} \in L_{\lowerSbound}^-(\certificate)\bigr) \wedge \bigl(\forall t < \theStoppingTime: \trajectoryProcess{\policy}{0}{x_0}{t} \notin \StatesUnsafe\bigr) \Bigr\} \geq \probRA.
  \]
  Because \(L_{\lowerSbound}^-(\certificate) \subset L_{\upperSbound}^-(\certificate) \subset \StatesTarget\) by construction and the goal condition, this event is a subset of the reach-avoid event \(E_{\text{RA}}\) of \eqref{eq:reach-avoid}; therefore, the reach-avoid property is satisfied with probability at least \(\probRA\).
  
  Next, we prove that the stay property is satisfied as well.
  
Using Proposition~\ref{lm:maximal} and shifting the time index by \(\theStoppingTime\) (which is possible because the process is Markovian),
  \[
    \probability \Bigl\{\sup_{t \geq \theStoppingTime} \certificate(t, \trajectoryProcess{\policy}{0}{x_0}{t}) \geq \upperSbound \Bigr\} \leq \tfrac{1}{\upperSbound}\expect\bigl[\certificate(0, \trajectoryProcess{\policy}{0}{x_0}{\theStoppingTime})\bigr] \leq \tfrac{\lowerSbound}{\upperSbound}.
  \]
  The opposite event can be written as
  \begin{align*}
    \Bigl\{\sup_{t \geq \theStoppingTime} &\certificate(t, \trajectoryProcess{\policy}{0}{x_0}{t}) < \upperSbound \Bigr\} = \\
    &\Bigl\{\forall t \geq \theStoppingTime: \certificate(t, \trajectoryProcess{\policy}{0}{x_0}{t}) < \upperSbound \Bigr\}\\
    &\cap \Bigl\{\nexists \gamma < \upperSbound: \forall t \geq \theStoppingTime: \certificate(t, \trajectoryProcess{\policy}{0}{x_0}{t}) < \gamma \Bigr\}.
  \end{align*}
  Because \(\probability [A] \geq \probability [A \cap B]\) for any \(A\) and \(B\), this implies
  \[
    \probability \Bigl\{\forall t \geq \theStoppingTime: \certificate(t, \trajectoryProcess{\policy}{0}{x_0}{t}) < \upperSbound \Bigr\} 
    \geq 1 - \tfrac{\lowerSbound}{\upperSbound} = \probS.
  \]
  Again, \(L_{\lowerSbound}^-(\certificate) \subseteq \StatesTarget\) due to the goal condition.
  Thus, this event is a subset of the stay event \(E_{\text{S}}\) of \eqref{eq:stay},  which means that the stay part of the specification is satisfied with probability at least \(\probS\).
\end{proof}

\subsection{Proof of Lemma~\ref{lm:rasm-is-supermartingale}}

Lemma~\ref{lm:rasm-is-supermartingale} is proven using the following proposition.

\begin{proposition}[first exit process is a supermartingale, cf. \citet{khasminskii2011stability}, Lemma 5.1]\label{pr:exit-supermartingale}
  Let \(\certificate(t,x)\) be a function twice continuously differentiable with respect to \(x\), continuously differentiable with respect to \(t\) on the set \(\NonNegativeReals \times \mathbb{U}\) for a bounded domain \(\mathbb{U}\).
  Moreover, in this set \(\generator_\policy \certificate \leq 0\).
  Let \(\tau_\mathbb{U}\) be the first exit time from \(\mathbb{U}\).
  Then the process \(\certificate(t \wedge \tau_{\mathbb{U}}, \trajectoryProcess{\policy}{0}{x_0}{t \wedge \tau_{\mathbb{U}}})\) is a supermartingale.
\end{proposition}

\begin{proof}[Proof of Lemma~\ref{lm:rasm-is-supermartingale}]
  Follows immediately from Proposition~\ref{pr:exit-supermartingale} and boundedness of \(\States\) and therefore any of its subsets.
\end{proof}

\subsection{Proof Sketch for Corollary~\ref{cor:ra-s-separately}}

Since both proofs are very similar to the proof of Theorem~\ref{th:main-result}, we do not present their full versions, but restrict ourselves to a proof sketch.

\begin{proof}[Sketch of a proof for Corollary~\ref{cor:ra-s-separately}]
  The case of staying is already part of the Theorem~\ref{th:main-result} proof.
The proof for reach-avoidance without staying follows the steps of Theorem~\ref{th:main-result} proof for a stopping time
\begin{equation*}
  \define{\theStoppingTime}{\inf_{t \geq 0}\bigl\{t \bigm| \certificate(t, \trajectoryProcess{\policy}{0}{x_0}{t}) \notin L^-_{\upperRAbound}(\certificate)\setminus \interior\StatesTargetFamily \bigr\}.}
\end{equation*}
It is also similar to the proof of discrete-time reach-avoidance done by \citet{vzikelic2023learning}.
\end{proof}

\subsection{Computing Infrastructure}

We conducted the experiments on MacBook Pro (Model: 14'' 2021, CPU: Apple M1 Max, RAM: 32~GB, OS: macOS Sonoma 14.5).
The experiments were run using Python 3.11.7.
The names and versions of the libraries we used are included in the code appendix in \texttt{requirements.txt}.

\subsection{Hyperparameter Values}

For the inverted pendulum, we use a policy pre-trained with \texttt{torchRL} for the deterministic version of the problem.
The script used for training and the saved policy are both included in the code appendix.
The policy network consists of two linear layers with hyperbolic tangent activations, followed by a final linear layer.
The hidden layers consist of 64 neurons.

For the certificate network, the architecture is the same, but with the addition of a softplus activation at the end to make the values nonnegative.
The hidden layers contain 32 neurons.

The values of the remaining hyperparameters are summarized in Table~\ref{tbl:hyper} and can be found in the code appendix.

\begin{table}[htb]
 \centering
 \caption{Hyperparameter values.}
 \begin{tabular}{lcrr}
 \toprule
 Parameter && GBM & Pendulum\\
 \midrule
  verification frequency & \(q\) & 1000 & 1000\\
  batch size & \(n\) & 256 & 256 \\
  verifier mesh size & \(m\) & 200 & 400 \\
  generator threshold & \(\generatorBoundRA\) & 1.0 & 1.0\\
  regularizer multiplier & \(\lambda\) & \(10^{-1}\) & \(10^{-1}\) \\
  verification slack & \(\kappa\) & 4 & 4\\
  maximum verification depth & \(k\) & 2 & 5\\
  optimizer && \texttt{Adam} & \texttt{Adam} \\
  optimizer learning rate && \(10^{-3}\) & \(10^{-3}\) \\
 \bottomrule
 \end{tabular}
 \label{tbl:hyper}
\end{table}

\subsection{Experimental Specifications}

Both experiments use RAS specifications with \(\probRA = 0.9\) and \(\probS=0.9\), and the sets defined as follows.

\subsubsection{Inverted Pendulum.} The dynamics of the deterministic inverted pendulum are given by the following equation:
\[\ddot\pndAngle = \tfrac{\pndGravity}{\pndLength}  \sin \pndAngle + \tfrac{\pndMaxTorque u - \pndFriction \dot\pndAngle}{\pndMass \pndLength^2},\]
where \(\pndAngle\) is the current angle, \(\pndGravity = 9.81\) is the standard acceleration due to gravity, \(\pndLength = 0.5\) is the pendulum length, \(\pndMass = 0.15\) is the ball mass, \(b = 0.1\) is the coefficient of friction, and \(\pndMaxTorque = 6\) is the maximum torque (so that \(\lvert u \rvert \leq 1\)).

The angular velocity \(\pndVelocity\) of the pendulum in the deterministic case is \(\pndVelocity = \dot\pndAngle\).
The agent observes both the angular velocity \(\pndVelocity\) and the angle \(\pndAngle\), that is, \(x = [\pndVelocity, \pndAngle]^\top\).

The specification sets are:
\begin{align*}
  \States &= [-20,20] \times [-2\pi, 2\pi], &
  \StatesInitial &= [-1, 1] \times [\tfrac{3\pi}{4}, \tfrac{5\pi}{4}],
\end{align*}
\[\StatesTarget = [-4, 4] \times [-\tfrac{\pi}{2}, \tfrac{\pi}{2}],\]
\[
  \StatesUnsafe = \bigl([-20, -10] \times [-2\pi, -\tfrac{3\pi}{2}]\bigr) \cup \bigl([10, 20] \times [\tfrac{3\pi}{2}, 2\pi]\bigr).
\]

\subsubsection{GBM.} The specification sets are:
\begin{align*}
  \States &= [-100,100]^2, &
  \StatesInitial &= [45, 55] \times [-55, -45],\\
  \StatesTarget &= [-25,25]^2, &
  \StatesUnsafe &= [-100,-80] \times [-100, 100].
\end{align*}

\section{Acknowledgements}

This research was funded by the NWO Veni grant Explainable Monitoring (222.119), and the Advanced Research + Invention Agency (ARIA) under the Safeguarded AI programme. This work was done in part while Anna Lukina was visiting the Simons Institute for the Theory of Computing.

\bibliography{Neustroev}

\begin{thebibliography}{57}
\providecommand{\natexlab}[1]{#1}

\bibitem[{Abate et~al.(2021)Abate, Ahmed, Giacobbe, and Peruffo}]{abate2020formal}
Abate, A.; Ahmed, D.; Giacobbe, M.; and Peruffo, A. 2021.
\newblock Formal Synthesis of {L}yapunov Neural Networks.
\newblock \emph{{IEEE} Control. Syst. Lett.}, 5(3): 773--778.

\bibitem[{Abate et~al.(2023)Abate, Edwards, Giacobbe, Punchihewa, and Roy}]{DBLP:conf/concur/AbateEGPR23}
Abate, A.; Edwards, A.; Giacobbe, M.; Punchihewa, H.; and Roy, D. 2023.
\newblock Quantitative Verification with Neural Networks.
\newblock In \emph{{CONCUR}}, volume 279 of \emph{LIPIcs}, 22:1--22:18. Schloss Dagstuhl - Leibniz-Zentrum f{\"{u}}r Informatik.

\bibitem[{Abate, Giacobbe, and Roy(2021)}]{abate2021learning}
Abate, A.; Giacobbe, M.; and Roy, D. 2021.
\newblock Learning Probabilistic Termination Proofs.
\newblock In \emph{{CAV} {(2)}}, volume 12760 of \emph{Lecture Notes in Computer Science}, 3--26. Springer.

\bibitem[{Abate, Giacobbe, and Roy(2024)}]{DBLP:conf/cav/AbateGR24}
Abate, A.; Giacobbe, M.; and Roy, D. 2024.
\newblock Stochastic Omega-Regular Verification and Control with Supermartingales.
\newblock In \emph{{CAV} {(3)}}, volume 14683 of \emph{Lecture Notes in Computer Science}, 395--419. Springer.

\bibitem[{Akintunde et~al.(2022)Akintunde, Botoeva, Kouvaros, and Lomuscio}]{DBLP:journals/aamas/AkintundeBKL22}
Akintunde, M.~E.; Botoeva, E.; Kouvaros, P.; and Lomuscio, A. 2022.
\newblock Formal verification of neural agents in non-deterministic environments.
\newblock \emph{Auton. Agents Multi Agent Syst.}, 36(1): 6.

\bibitem[{Anand and Zamani(2023)}]{ANAND20232431}
Anand, M.; and Zamani, M. 2023.
\newblock Formally Verified Neural Network Control Barrier Certificates for Unknown Systems.
\newblock \emph{IFAC-PapersOnLine}, 56(2): 2431--2436.

\bibitem[{Ansaripour et~al.(2023)Ansaripour, Chatterjee, Henzinger, Lechner, and {\v Z}ikeli\'c}]{ansaripour2023learning}
Ansaripour, M.; Chatterjee, K.; Henzinger, T.~A.; Lechner, M.; and {\v Z}ikeli\'c, D. 2023.
\newblock Learning Provably Stabilizing Neural Controllers for Discrete-Time Stochastic Systems.
\newblock In \emph{{ATVA} {(1)}}, volume 14215 of \emph{Lecture Notes in Computer Science}, 357--379. Springer.

\bibitem[{Bacci, Giacobbe, and Parker(2021)}]{DBLP:conf/ijcai/BacciG021}
Bacci, E.; Giacobbe, M.; and Parker, D. 2021.
\newblock Verifying Reinforcement Learning up to Infinity.
\newblock In \emph{{IJCAI}}, 2154--2160. ijcai.org.

\bibitem[{Badings et~al.(2024)Badings, Koops, Junges, and Jansen}]{badings2024learning}
Badings, T.; Koops, W.; Junges, S.; and Jansen, N. 2024.
\newblock Learning-Based Verification of Stochastic Dynamical Systems with Neural Network Policies.
\newblock arXiv:2406.00826.

\bibitem[{Blumenthal and Getoor(1968)}]{blumenthal1968markov}
Blumenthal, R.; and Getoor, R. 1968.
\newblock \emph{Markov Processes and Potential Theory}.
\newblock Pure and applied mathematics : a series of monographs and textbooks. Academic Press.
\newblock ISBN 9780121078508.

\bibitem[{Chang, Roohi, and Gao(2019)}]{chang2019neural}
Chang, Y.-C.; Roohi, N.; and Gao, S. 2019.
\newblock Neural {L}yapunov Control.
\newblock In \emph{{NeurIPS}}, 3240--3249.

\bibitem[{Chatterjee et~al.(2023)Chatterjee, Henzinger, Lechner, and {\v Z}ikeli\'c}]{DBLP:conf/tacas/ChatterjeeHLZ23}
Chatterjee, K.; Henzinger, T.~A.; Lechner, M.; and {\v Z}ikeli\'c, D. 2023.
\newblock A Learner-Verifier Framework for Neural Network Controllers and Certificates of Stochastic Systems.
\newblock In \emph{{TACAS} {(1)}}, volume 13993 of \emph{Lecture Notes in Computer Science}, 3--25. Springer.

\bibitem[{Chen et~al.(2021{\natexlab{a}})Chen, Fazlyab, Morari, Pappas, and Preciado}]{DBLP:conf/hybrid/ChenFMPP21}
Chen, S.; Fazlyab, M.; Morari, M.; Pappas, G.~J.; and Preciado, V.~M. 2021{\natexlab{a}}.
\newblock Learning {L}yapunov functions for hybrid systems.
\newblock In \emph{{HSCC}}, 13:1--13:11. {ACM}.

\bibitem[{Chen et~al.(2021{\natexlab{b}})Chen, Fazlyab, Morari, Pappas, and Preciado}]{chen2021learning}
Chen, S.; Fazlyab, M.; Morari, M.; Pappas, G.~J.; and Preciado, V.~M. 2021{\natexlab{b}}.
\newblock Learning Region of Attraction for Nonlinear Systems.
\newblock In \emph{{CDC}}, 6477--6484. {IEEE}.

\bibitem[{Dawson, Gao, and Fan(2023)}]{dawson2023safe}
Dawson, C.; Gao, S.; and Fan, C. 2023.
\newblock Safe Control With Learned Certificates: {A} Survey of Neural {L}yapunov, Barrier, and Contraction Methods for Robotics and Control.
\newblock \emph{{IEEE} Trans. Robotics}, 39(3): 1749--1767.

\bibitem[{Dawson et~al.(2021)Dawson, Qin, Gao, and Fan}]{dawson2022safe}
Dawson, C.; Qin, Z.; Gao, S.; and Fan, C. 2021.
\newblock Safe Nonlinear Control Using Robust Neural {L}yapunov-Barrier Functions.
\newblock In \emph{CoRL}, volume 164 of \emph{Proceedings of Machine Learning Research}, 1724--1735. {PMLR}.

\bibitem[{Edwards, Peruffo, and Abate(2023)}]{edwards4880686general}
Edwards, A.; Peruffo, A.; and Abate, A. 2023.
\newblock A General Verification Framework for Dynamical and Control Models via Certificate Synthesis.
\newblock \emph{CoRR}, abs/2309.06090.

\bibitem[{Gehr et~al.(2018)Gehr, Mirman, Drachsler-Cohen, Tsankov, Chaudhuri, and Vechev}]{DBLP:conf/sp/GehrMDTCV18}
Gehr, T.; Mirman, M.; Drachsler-Cohen, D.; Tsankov, P.; Chaudhuri, S.; and Vechev, M.~T. 2018.
\newblock {AI2:} Safety and Robustness Certification of Neural Networks with Abstract Interpretation.
\newblock In \emph{{IEEE} Symposium on Security and Privacy}, 3--18. {IEEE} Computer Society.

\bibitem[{Giacobbe et~al.(2024)Giacobbe, Kroening, Pal, and Tautschnig}]{neuralmc}
Giacobbe, M.; Kroening, D.; Pal, A.; and Tautschnig, M. 2024.
\newblock Neural Model Checking.
\newblock In \emph{{NeurIPS}}.

\bibitem[{Giacobbe, Kroening, and Parsert(2022)}]{DBLP:conf/sigsoft/GiacobbeKP22}
Giacobbe, M.; Kroening, D.; and Parsert, J. 2022.
\newblock Neural termination analysis.
\newblock In \emph{{ESEC/SIGSOFT} {FSE}}, 633--645. {ACM}.

\bibitem[{Kalman and Bertram(1959)}]{kalman1960control}
Kalman, R.~E.; and Bertram, J.~E. 1959.
\newblock Control system analysis and design via the second method of {L}yapunov: (I) continuous-time systems (II) discrete time systems.
\newblock \emph{IRE Transactions on Automatic Control}, 4(3): 112.

\bibitem[{Katz et~al.(2017)Katz, Barrett, Dill, Julian, and Kochenderfer}]{DBLP:conf/cav/KatzBDJK17}
Katz, G.; Barrett, C.~W.; Dill, D.~L.; Julian, K.; and Kochenderfer, M.~J. 2017.
\newblock Reluplex: An Efficient {SMT} Solver for Verifying Deep Neural Networks.
\newblock In \emph{{CAV} {(1)}}, volume 10426 of \emph{Lecture Notes in Computer Science}, 97--117. Springer.

\bibitem[{Khasminskii(2011)}]{khasminskii2011stability}
Khasminskii, R. 2011.
\newblock \emph{Stochastic Stability of Differential Equations}.
\newblock Stochastic Modelling and Applied Probability. Springer.

\bibitem[{Kolter and Manek(2019)}]{kolter2019learning}
Kolter, J.~Z.; and Manek, G. 2019.
\newblock Learning Stable Deep Dynamics Models.
\newblock In \emph{{NeurIPS}}, 11126--11134.

\bibitem[{Kouvaros and Lomuscio(2021)}]{kouvaros2021towards}
Kouvaros, P.; and Lomuscio, A. 2021.
\newblock Towards Scalable Complete Verification of {ReLU} Neural Networks via Dependency-based Branching.
\newblock In \emph{IJCAI}, 2643--2650.

\bibitem[{Le~Gall(2016)}]{gall2016brownian}
Le~Gall, J.-F. 2016.
\newblock \emph{{B}rownian Motion, Martingales, and Stochastic Calculus}.
\newblock Graduate Texts in Mathematics. Springer.

\bibitem[{Li et~al.(2020)Li, Wong, Chen, and Duvenaud}]{li2020scalable}
Li, X.; Wong, T.-K.~L.; Chen, R. T.~Q.; and Duvenaud, D. 2020.
\newblock Scalable Gradients for Stochastic Differential Equations.
\newblock In \emph{{AISTATS}}, volume 108 of \emph{Proceedings of Machine Learning Research}, 3870--3882. {PMLR}.

\bibitem[{Long and Bayoumi(1993)}]{long1993feedback}
Long, Y.; and Bayoumi, M. 1993.
\newblock Feedback stabilization: control {L}yapunov functions modelled by neural networks.
\newblock In \emph{{CDC}}, 2812--2814. IEEE.

\bibitem[{Mandal et~al.(2024)Mandal, Amir, Wu, Daukantas, Newell, Ravaioli, Meng, Durling, Ganai, Shim, Katz, and Barrett}]{mandal2024formally}
Mandal, U.; Amir, G.; Wu, H.; Daukantas, I.; Newell, F.~L.; Ravaioli, U.~J.; Meng, B.; Durling, M.; Ganai, M.; Shim, T.; Katz, G.; and Barrett, C. 2024.
\newblock Formally Verifying Deep Reinforcement Learning Controllers with {L}yapunov Barrier Certificates.
\newblock arXiv:2405.14058.

\bibitem[{Mathiesen, Calvert, and Laurenti(2023)}]{DBLP:journals/csysl/MathiesenCL23}
Mathiesen, F.~B.; Calvert, S.~C.; and Laurenti, L. 2023.
\newblock Safety Certification for Stochastic Systems via Neural Barrier Functions.
\newblock \emph{{IEEE} Control. Syst. Lett.}, 7: 973--978.

\bibitem[{Nadali et~al.(2024)Nadali, Murali, Trivedi, and Zamani}]{DBLP:conf/aaai/NadaliM0024}
Nadali, A.; Murali, V.; Trivedi, A.; and Zamani, M. 2024.
\newblock Neural Closure Certificates.
\newblock In \emph{{AAAI}}, 21446--21453. {AAAI} Press.

\bibitem[{Noroozi et~al.(2008)Noroozi, Karimaghaee, Safaei, and Javadi}]{noroozi2008generation}
Noroozi, N.; Karimaghaee, P.; Safaei, F.; and Javadi, H. 2008.
\newblock Generation of {L}yapunov functions by neural networks.
\newblock In \emph{Proceedings of the World Congress on Engineering}.

\bibitem[{Papachristodoulou and Prajna(2002)}]{papachristodoulou2002construction}
Papachristodoulou, A.; and Prajna, S. 2002.
\newblock On the construction of {L}yapunov functions using the sum of squares decomposition.
\newblock In \emph{{CDC}}, 3482--3487. {IEEE}.

\bibitem[{Petridis and Petridis(2006)}]{DBLP:conf/ijcnn/PetridisP06}
Petridis, V.; and Petridis, S. 2006.
\newblock Construction of Neural Network Based {L}yapunov Functions.
\newblock In \emph{{IJCNN}}, 5059--5065. {IEEE}.

\bibitem[{Prajna, Jadbabaie, and Pappas(2004)}]{prajna2004stochastic}
Prajna, S.; Jadbabaie, A.; and Pappas, G.~J. 2004.
\newblock Stochastic safety verification using barrier certificates.
\newblock In \emph{{CDC}}, 929--934. {IEEE}.

\bibitem[{Prajna, Jadbabaie, and Pappas(2007)}]{prajna2007framework}
Prajna, S.; Jadbabaie, A.; and Pappas, G.~J. 2007.
\newblock A Framework for Worst-Case and Stochastic Safety Verification Using Barrier Certificates.
\newblock \emph{{IEEE} Trans. Autom. Control.}, 52(8): 1415--1428.

\bibitem[{Prokhorov(1994)}]{prokhorov1994lyapunov}
Prokhorov, D.~V. 1994.
\newblock A {L}yapunov machine for stability analysis of nonlinear systems.
\newblock In \emph{{ICNN}}, 1028--1031. IEEE.

\bibitem[{Qin et~al.(2021)Qin, Zhang, Chen, Chen, and Fan}]{DBLP:conf/iclr/QinZCCF21}
Qin, Z.; Zhang, K.; Chen, Y.; Chen, J.; and Fan, C. 2021.
\newblock Learning Safe Multi-agent Control with Decentralized Neural Barrier Certificates.
\newblock In \emph{{ICLR}}. OpenReview.net.

\bibitem[{Richards, Berkenkamp, and Krause(2018)}]{richards2018lyapunov}
Richards, S.~M.; Berkenkamp, F.; and Krause, A. 2018.
\newblock The {L}yapunov Neural Network: Adaptive Stability Certification for Safe Learning of Dynamical Systems.
\newblock In \emph{CoRL}, volume~87 of \emph{Proceedings of Machine Learning Research}, 466--476. {PMLR}.

\bibitem[{Serpen(2005)}]{serpen2005empirical}
Serpen, G. 2005.
\newblock Empirical approximation for {L}yapunov functions with artificial neural nets.
\newblock In \emph{{IJCNN}}, 735--740. IEEE.

\bibitem[{Singla and Feizi(2020)}]{singla2020secondorder}
Singla, S.; and Feizi, S. 2020.
\newblock Second-Order Provable Defenses against Adversarial Attacks.
\newblock In \emph{{ICML}}, volume 119 of \emph{Proceedings of Machine Learning Research}, 8981--8991. {PMLR}.

\bibitem[{Stein and Shakarchi(2009)}]{stein2009real}
Stein, E.~M.; and Shakarchi, R. 2009.
\newblock \emph{Real analysis: measure theory, integration, and {H}ilbert spaces}.
\newblock Princeton University Press.

\bibitem[{Sun, Jha, and Fan(2020)}]{DBLP:conf/corl/SunJF20}
Sun, D.; Jha, S.; and Fan, C. 2020.
\newblock Learning Certified Control Using Contraction Metric.
\newblock In \emph{CoRL}, volume 155 of \emph{Proceedings of Machine Learning Research}, 1519--1539. {PMLR}.

\bibitem[{Takeishi and Kawahara(2021)}]{takeishi2021learning}
Takeishi, N.; and Kawahara, Y. 2021.
\newblock Learning Dynamics Models with Stable Invariant Sets.
\newblock In \emph{{AAAI}}, 9782--9790. {AAAI} Press.

\bibitem[{Topcu, Packard, and Seiler(2008)}]{topcu2008local}
Topcu, U.; Packard, A.~K.; and Seiler, P.~J. 2008.
\newblock Local stability analysis using simulations and sum-of-squares programming.
\newblock \emph{Autom.}, 44(10): 2669--2675.

\bibitem[{Tran et~al.(2020)Tran, Yang, Lopez, Musau, Nguyen, Xiang, Bak, and Johnson}]{DBLP:conf/cav/TranYLMNXBJ20}
Tran, H.-D.; Yang, X.; Lopez, D.~M.; Musau, P.; Nguyen, L.~V.; Xiang, W.; Bak, S.; and Johnson, T.~T. 2020.
\newblock {NNV:} The Neural Network Verification Tool for Deep Neural Networks and Learning-Enabled Cyber-Physical Systems.
\newblock In \emph{{CAV} {(1)}}, volume 12224 of \emph{Lecture Notes in Computer Science}, 3--17. Springer.

\bibitem[{Wu et~al.(2024)Wu, Isac, Zeljic, Tagomori, Daggitt, Kokke, Refaeli, Amir, Julian, Bassan, Huang, Lahav, Wu, Zhang, Komendantskaya, Katz, and Barrett}]{DBLP:conf/cav/WuIZTDKRAJBHLWZKKB24}
Wu, H.; Isac, O.; Zeljic, A.; Tagomori, T.; Daggitt, M.~L.; Kokke, W.; Refaeli, I.; Amir, G.; Julian, K.; Bassan, S.; Huang, P.; Lahav, O.; Wu, M.; Zhang, M.; Komendantskaya, E.; Katz, G.; and Barrett, C.~W. 2024.
\newblock Marabou 2.0: {A} Versatile Formal Analyzer of Neural Networks.
\newblock In \emph{{CAV} {(2)}}, volume 14682 of \emph{Lecture Notes in Computer Science}, 249--264. Springer.

\bibitem[{Wu et~al.(2023)Wu, Clark, Kantaros, and Vorobeychik}]{wu2023neural}
Wu, J.; Clark, A.; Kantaros, Y.; and Vorobeychik, Y. 2023.
\newblock Neural {L}yapunov Control for Discrete-Time Systems.
\newblock In \emph{{NeurIPS}}.

\bibitem[{Xu et~al.(2020)Xu, Shi, Zhang, Wang, Chang, Huang, Kailkhura, Lin, and Hsieh}]{xu2020automatic}
Xu, K.; Shi, Z.; Zhang, H.; Wang, Y.; Chang, K.-W.; Huang, M.; Kailkhura, B.; Lin, X.; and Hsieh, C.-J. 2020.
\newblock Automatic Perturbation Analysis for Scalable Certified Robustness and Beyond.
\newblock In \emph{{NeurIPS}}.

\bibitem[{Yang et~al.(2024)Yang, Dai, Shi, Hsieh, Tedrake, and Zhang}]{yanglyapunov}
Yang, L.; Dai, H.; Shi, Z.; Hsieh, C.-J.; Tedrake, R.; and Zhang, H. 2024.
\newblock {L}yapunov-Stable Neural Control for State and Output Feedback: {A} Novel Formulation.
\newblock In \emph{{ICML}}. OpenReview.net.

\bibitem[{Zhang et~al.(2018)Zhang, Weng, Chen, Hsieh, and Daniel}]{zhang2018efficient}
Zhang, H.; Weng, T.-W.; Chen, P.-Y.; Hsieh, C.-J.; and Daniel, L. 2018.
\newblock Efficient Neural Network Robustness Certification with General Activation Functions.
\newblock In \emph{{NeurIPS}}, 4944--4953.

\bibitem[{Zhang et~al.(2023{\natexlab{a}})Zhang, Wu, Vorobeychik, and Clark}]{DBLP:conf/nips/Zhang0V023}
Zhang, H.; Wu, J.; Vorobeychik, Y.; and Clark, A. 2023{\natexlab{a}}.
\newblock Exact Verification of {ReLU} Neural Control Barrier Functions.
\newblock In \emph{{NeurIPS}}.

\bibitem[{Zhang, Zhu, and Lin(2022)}]{zhang2022neural}
Zhang, J.; Zhu, Q.; and Lin, W. 2022.
\newblock Neural Stochastic Control.
\newblock In \emph{{NeurIPS}}.

\bibitem[{Zhang et~al.(2023{\natexlab{b}})Zhang, Xiu, Qu, and Fan}]{DBLP:conf/l4dc/ZhangXQF23}
Zhang, S.; Xiu, Y.; Qu, G.; and Fan, C. 2023{\natexlab{b}}.
\newblock Compositional Neural Certificates for Networked Dynamical Systems.
\newblock In \emph{{L4DC}}, volume 211 of \emph{Proceedings of Machine Learning Research}, 272--285. {PMLR}.

\bibitem[{Zhao et~al.(2020)Zhao, Zeng, Chen, and Liu}]{DBLP:conf/hybrid/ZhaoZC020}
Zhao, H.; Zeng, X.; Chen, T.; and Liu, Z. 2020.
\newblock Synthesizing barrier certificates using neural networks.
\newblock In \emph{{HSCC}}, 25:1--25:11. {ACM}.

\bibitem[{{\v Z}ikeli\'c et~al.(2023{\natexlab{a}}){\v Z}ikeli\'c, Lechner, Henzinger, and Chatterjee}]{vzikelic2023learning}
{\v Z}ikeli\'c, D.; Lechner, M.; Henzinger, T.~A.; and Chatterjee, K. 2023{\natexlab{a}}.
\newblock Learning Control Policies for Stochastic Systems with Reach-Avoid Guarantees.
\newblock In \emph{{AAAI}}, 11926--11935. {AAAI} Press.

\bibitem[{{\v Z}ikeli\'c et~al.(2023{\natexlab{b}}){\v Z}ikeli\'c, Lechner, Verma, Chatterjee, and Henzinger}]{vzikelic2024compositional}
{\v Z}ikeli\'c, D.; Lechner, M.; Verma, A.; Chatterjee, K.; and Henzinger, T.~A. 2023{\natexlab{b}}.
\newblock Compositional Policy Learning in Stochastic Control Systems with Formal Guarantees.
\newblock In \emph{{NeurIPS}}.

\end{thebibliography}

\end{document}